\def\llncsStyle{0}
\def\fullV{1}

\ifnum\llncsStyle=1
\documentclass{llncs}
\usepackage{amsmath,amsfonts,amssymb,color}
\else
\documentclass[11pt]{article}
\usepackage[letterpaper,hmargin=1in,vmargin=1.25in]{geometry}
\usepackage{amsmath,amsfonts,amssymb,amsthm,url,color}
\fi
\usepackage{algorithm,epsfig,comment}
\usepackage[usenames,dvipsnames]{pstricks}

\ifnum\llncsStyle=1

\newcommand{\myqed}{\qed}
\else

\newcommand{\myqed}{}
\fi

\ifnum\llncsStyle=1
\spnewtheorem{fact}[theorem]{Fact}{\bfseries}{\itshape}
\spnewtheorem{assumption}[theorem]{Assumption}{\bfseries}{\itshape}
\spnewtheorem*{sketch}{Proof sketch}{\itshape}{}
\spnewtheorem{construction}[theorem]{Construction}{\bfseries}{\itshape}
\spnewtheorem{reduction}[theorem]{Reduction}{\bfseries}{\itshape}
\spnewtheorem{myclaim}[theorem]{Claim}{\bfseries}{\itshape}

\else
\newtheorem{theorem}{Theorem}[section]
\newtheorem{lemma}[theorem]{Lemma}

\newtheorem{corollary}[theorem]{Corollary}

\newtheorem{myclaim}[theorem]{Claim}
\newtheorem{definition}[theorem]{Definition}
\newtheorem{remark}[theorem]{Remark}

\newtheorem{observation}[theorem]{Observation}
\newtheorem{construction}[theorem]{Construction}

\fi

\newcommand{\marginnote}[1]{}

\newcommand{\half}{\frac{1}{2}}

\newcommand{\eps}{\varepsilon}
\newcommand{\N}{\mathbb{N}}
\newcommand{\RR}{\mathbb{R}}

\newcommand{\poly}{{\rm poly}}

\newcommand{\noise}{\mathsf{noise}}
\newcommand{\bit}{\{0,1\}}
\newcommand{\str}{\bit^*}

\newcommand{\BIT}[1]{\bit^{#1}}
\newcommand{\nbit}{\bit^n}

\newcommand{\ra}{\rightarrow}

\newcommand{\set}[1]{\left\{ #1\right\}}

\providecommand{\drawn}[1]{\stackrel{#1}{\leftarrow}}
\providecommand{\rs}{\drawn{R}}

\newcommand{\GF}{\mathrm{GF}}

\ifnum\fullV=1
\newcommand{\full}[1]{#1}
\newcommand{\short}[1]{}
\else
\newcommand{\short}[1]{#1}
\newcommand{\full}[1]{}
\fi
\newcommand{\fullfull}[1]{}
\newcommand{\remove}[1]{}
\newcommand{\bnote}[1]{}
\newcommand{\gnote}[1]{}

\newcommand{\enc}{\mathsf{Enc}}
\newcommand{\dec}{\mathsf{Dec}}
\newcommand{\BSC}{\mathsf{BSC}}

\newcommand{\er}{\mathsf{err}}

\newcommand{\out}{\mathsf{out}}
\newcommand{\inn}{\mathsf{in}}
\newcommand{\Cout}{C_{\out}}
\newcommand{\Cin}{C_{\inn}}
\newcommand{\wt}{\mathsf{wt}}
\newcommand{\nout}{n_{\out}}
\newcommand{\kout}{k_{\out}}
\newcommand{\rout}{R_{\out}}
\newcommand{\nin}{n_{\inn}}
\newcommand{\kin}{k_{\inn}}
\newcommand{\Lin}{L_{\inn}}

\newcommand{\dist}{\mathsf{dist}}

\newcommand{\Prob}{\Pr}
\newcommand{\Pikn}{\Pi_{k,n}}

\newcommand{\GV}{\text{\sc gv}}

\newcommand{\bin}{\text{\emph{bin}}}

\newcommand{\argmax}{\operatornamewithlimits{arg\ max}}

\begin{document}

\ifnum\llncsStyle=1
\title{Deterministic Rateless Codes for BSC}
\author{Benny Applebaum \inst{} 
\and Liron David \inst{} \and Guy Even \inst{}
}
\institute{School of Electrical Engineering, Tel-Aviv University\\ \email{\{bennyap,lirondav,guy\}@post.tau.ac.il}}
\else
\title{Deterministic Rateless Codes for the Binary Symmetric Channel}
\author{Benny Applebaum\thanks{School of Electrical Engineering, Tel-Aviv University, \texttt{\{bennyap,lirondav,guy\}@post.tau.ac.il}.} \and Liron David\footnotemark[1] \and Guy Even\footnotemark[1]}
\fi

\date{\today}
\maketitle

\thispagestyle{empty}

\begin{abstract}
  A rateless code encodes a finite length information word into an infinitely long
  codeword such that longer prefixes of the codeword can tolerate a larger fraction of errors.
  A rateless code
  achieves capacity for a family of channels if, for every channel in the family,
  reliable communication is obtained by a prefix of the code whose rate is arbitrarily close to the
  channel's capacity. As a result, a universal encoder 
  can communicate over all channels in the family while simultaneously achieving optimal communication overhead.

  In this paper, we construct the first \emph{deterministic} rateless code for the
  binary symmetric channel. Our code can be encoded and decoded in $O(\beta)$ time
  per bit and in almost logarithmic parallel time of $O(\beta \log n)$, where $\beta$ is any (arbitrarily slow) super-constant function.
  Furthermore, the error probability
  of our code is almost exponentially small $\exp(-\Omega(n/\beta))$.
  Previous rateless codes are probabilistic (i.e., based on code ensembles), require
  polynomial time per bit for decoding, and have inferior asymptotic error probabilities.

  Our main technical contribution is a constructive proof for the existence of an
  infinite generating matrix that each of its prefixes induce a weight distribution
  that approximates the expected weight distribution of a random linear code.
\end{abstract}
\newpage
\setcounter{page}{1}
\section{Introduction}
Consider a single transmitter $T$ who wishes to broadcast an information word $m\in
\bit^k$ to multiple receivers $B_1,\ldots,B_t$ over a Binary Symmetric Channel (BSC)
with crossover probability $p$.  By Shannon's theorem, using error correcting codes it is
possible to solve this problem with asymptotically optimal communication of $k\cdot
\frac{1}{C(p)-\delta}$ bits where $C(p)$ is the capacity of the channel and
$\delta>0$ is an arbitrarily small constant.  Furthermore, 
there are explicit capacity-achieving codes in which decoding and encoding can be
performed efficiently in polynomial or even linear time, e.g.~\cite{Barg00linear,Barg02error,Barg04error}.

The task of noisy broadcast becomes more challenging when each receiver $B_i$ experiences a different level of noise $p_i$ (e.g., due to a different distance from the transmitter).
Naively, one would use a code which is tailored to the noisiest channel with parameter $p_{\max}$. However, this will add an unnecessary communication overhead for receivers with lower noise level.
To make things worse, the transmitter may be unaware of the noise parameters, and, in some cases, may not even have a non-trivial upper-bound on the noise level.
Under these circumstances, the naive solution is not only wasteful but simply not applicable.

This problem (also studied in~\cite{ByersLMR98,shulman2000static}) can be solved by a \emph{rateless code}. Such a code  allows the transmitter to map the
information   word  $m\in  \bit^k$ into    an    infinitely long   sequence  of  bits
$\set{c_i}_{i\in  \N}$ such   that the longer the prefix   of the   codeword, the higher level of noise can be corrected. Ideally, we would    like to simultaneously   achieve the
optimal rate with respect   to all the noise  parameters  $p_i$.  That is, for  every
value of $p_i$, a prefix of length  $k\cdot \frac{1}{C(p_i)-\delta}$ should guarantee
reliable communication.

Rateless codes were extensively studied under various
names
\full{~\cite{mandelbaum1974adaptive,lin1984automatic,chase1985code,hagenauer1988rate,ByersLMR98,shulman2000static,rowitch2000performance,
caire2001throughput,ha2004rate,sesia2004incremental,ji2005rate,rajwan2007method,rajwan2008data}}. \short{(See~\cite{erez2012rateless} and references therein.)}
Information-theoretically, the problem of rateless
transmission is well understood~\cite{shulman2003communication}, and, for many noise models, random
codes provide an excellent (inefficient) solution.  The task of constructing
efficient rateless codes, which provide polynomial-time encoding and decoding, is
much more challenging.  Currently, only a few examples of efficient
capacity-achieving rateless codes are known for several important cases such as erasure channels, Gaussian channels, and binary symmetric channels~\cite{luby2002lt,shokrollahi2006raptor,erez2012rateless,perry2012spinal}.
Interestingly, all known constructions are probabilistic. Namely, the encoding
algorithm employs some public randomness, which is shared by the transmitter and all the receivers. (Equivalently, these constructions can be viewed as \emph{ensembles} of rateless codes.) This raises the natural question of whether randomness is inherently needed for rateless codes.\footnote{As we will see in Section~\ref{sec:overview}, the question is non-trivial even for computationally unbounded encoders as a rateless code is an infinite object.}
\vspace{-7pt}
\subsection{Our Results}
In this paper, we answer the question to the affirmative by constructing deterministic efficient rateless codes which achieve the capacity over the binary symmetric channel.
Letting $C(p)$ denote the capacity of the BSC with crossover probability $p$, we prove the following theorem.

\begin{theorem}[Main theorem]\label{thm:main}
Fix some super-constant function $\beta(k)=\omega(1)$.
There exists a deterministic rateless encoding algorithm $\enc$ and a deterministic rateless decoding algorithm $\dec$ with the following properties:
\begin{itemize}
  \item (\textbf{Capacity achieving}) For every 
  information word $m\in \bit^k$, noise parameter $p\in (0,\half)$, and prefix length $n=k\cdot\frac{1}{C(p)-\delta}$ where
  $0<\delta<C(p)$ is an arbitrary constant,
  we have that
  \[\Pr_{\noise\rs \BSC(p)}[\dec(\enc(m,[1:n])+\noise)\neq m]\leq 2^{-\Omega(k/\beta)},\]
  where $\enc(m,[1:n])$ denotes the $n$-bit prefix of the codeword $\enc(m)$, and the constants in the big Omega notation depend on $\delta$ and $p$.
  \item (\textbf{Efficiency}) The $n$-long prefix of $\enc$ can be computed in time $n \cdot \beta$, and decoding is performed in time $n \cdot \beta$. Both algorithms can be implemented in parallel by circuits of depth $O(\beta +\log n)$.
\end{itemize}
\end{theorem}
Letting $\beta$ be a slowly increasing function, (e.g, $\log^*(k)$) we obtain an ``almost'' exponential error and ``almost'' linear time encoding and decoding. 

One may also consider a weaker form of capacity achieving rateless codes in which the encoding is allowed to depend on the gap to capacity $\delta$.
(This effectively puts an a-priory upper-bound on the noise probability which makes things easier.)
In this setting we can obtain an asymptotically optimal construction with linear time encoding and decoding and exponentially small error.
\begin{theorem}\label{thm:weak}
For every $\delta>0$, there exists a deterministic encoding algorithm $\enc_{\delta}$ and a deterministic decoding algorithm $\dec_{\delta}$ with the following properties:
\begin{itemize}
  \item (\textbf{Weak capacity achieving}) For every
  information word $m\in \bit^k$, noise parameter $p\in (0,\half)$ such that
  $C(p)>\delta$,
and prefix length $n=k\cdot\frac{1}{C(p)-\delta}$
  we have that
  \[\Pr_{\noise\rs \BSC(p)}[\dec_{\delta}(\enc_{\delta}(m,[1:n])+\noise)\neq m]\leq 2^{-\Omega(k)}.\]
  \item (\textbf{Efficiency}) The $n$-long prefix of the code can be encoded and decoded in linear time $O(n )$ and in parallel by circuits of logarithmic depth $O(\log(n))$.
\end{itemize}
(The constants in the asymptotic notations depend on $\delta$.)
\end{theorem}

\paragraph{Comparison to Spinal codes.} Prior to our work, \emph{Spinal
  codes}~\cite{perry2011rateless,perry2012spinal,spinal2012} were the only known efficient (randomized) rateless
codes for the BSC. Apart from being deterministic, our construction has several
important theoretical advantages over spinal codes. The upper bound on the decoding error of spinal
codes is only inverse polynomial in $k$, and these codes only weakly achieve the
capacity (i.e., the encoding depends on the gap $\delta$ to capacity). Moreover, the
decoding complexity is polynomial 
(as opposed to linear or quasilinear in our codes), and both encoding and decoding are highly
sequential as they require $\Omega(k)$ sequential steps. \bnote{Linear and systematic.}  It
should be mentioned however that, while Spinal codes were reported to be highly practical, we
currently do not know whether our codes perform well in practice.

\subsection{Overview of our construction}\label{sec:overview}
Our starting point is a simple (yet inefficient and randomized) construction based on a random linear code.
Assume that both the encoder and decoder have an access to an infinite sequence of random $k$-bit row vectors $\set{R_i}_{i\in \N}$.
To encode the message $m\in \bit^k$, viewed as a $k$-bit column vector, the encoder sends the sequence $\set{R_i\cdot m}_{i\in \N}$ of inner products over the binary field.
To decode a noisy $n$-bit prefix of the codeword, we will employ the maximum-likelihood decoder (ML) for the code generated by the $n\times k$ matrix $R=(R_1,\ldots,R_n)$. A classical result in coding theory asserts that such a code achieves the capacity of the BSC.
Namely,
as long as the gap from capacity $\delta=C(p)-k/n$ is positive, the decoding error probability
\begin{equation}\label{Eq:ML}
  \Pr_{
  \noise\rs \BSC(p),R\rs \bit^{n\times k
  }}[\text{ML}_R(R\cdot m+\noise)\neq m]
\end{equation}
decreases exponentially fast as a function of $k$.

This construction has two important drawbacks: It is probabilistic and it does not support efficient decoding.
For now, let us ignore computational limitations, and attempt to de-randomize the construction.

\subsubsection{Derandomization}
We would like to deterministically generate an infinite number of rows $\set{R_i}_{i\in \N}$ such that
every $n$-row prefix matrix $R[1:n]=(R_1,\ldots,R_n)$ has a low ML-decoding error of,
say $0.01$, for every $p$ for which $C(p)-k/n$ is larger than, say,
$0.01$.\footnote{We use small constants to simplify the presentation, the discussion remains valid when the constants are replaced with a function that decreases with $k$.}

Although we know that, for every $n$, almost all $n\times k$ matrices
satisfy this condition, it is not a-priory clear that every such low-error matrix can
be extended to a larger matrix while preserving low error.

To solve this problem, we identify a property of \emph{good} matrices which, on one hand, guarantees low decoding error, and, on the other hand, is \emph{extendible} in the sense that every good matrix can be augmented by some row while preserving its goodness. We will base our notion of goodness on the \emph{weight distribution} of the matrix $R$.

Let $W_{i,n}$ denote the set of information words which are mapped by the matrix $R[1:n]$ to codewords of Hamming weight $i$, and let $w_{i,n}$ denote the size of this set. The sets $(W_{1,n},\ldots,W_{n,n})$ form a partition of $\bit^k$, and the vector $(w_{i,n})_{i=1,\ldots,n}$ is called the weight distribution of the code. When a row $R_{n+1}$ is added, the weight of all information words which are orthogonal to $R_{n+1}$ remains the same, while the weight of non-orthogonal words grows by 1. Thus $R_{n+1}$ splits $W_{i,n}$ to two parts: the orthogonal
vectors which ``remain'' in $W_{i,n+1}$, and the non-orthogonal vectors which
are ``elevated'' to $W_{i+1,n+1}$. A random row $R_{n+1}$ is therefore expected to split $W_{i,n}$ into two \emph{equal} parts.

If in each step we could choose such an ``ideal'' row which simultaneously halves all
$W_{i,n}$'s, we would get an ``ideal'' weight distribution in which
$w^*_i(n,k)=\binom{n}{i} \cdot 2^{k-n}$, as expected in a random linear code. Such an
ideal weight distribution guarantees a low ML decoding error over $\BSC(p)$ when
$C(p)<k/n$
(cf.~\cite{poltyrev1994bounds,shulman1999random,barg2002random}).

While we do not know how to choose such an ideal row (in fact it is
not clear that such a row exists), a probabilistic argument shows that
we can always find a row $R_{n+1}$ which approximately splits every
sufficiently large $W_{i,n}$ simultaneously. Furthermore, by keeping track of the small sets and choosing
$R_{n+1}$ which elevates a constant fraction of the lightest vectors, we make sure that the distance of the code is not
too small, e.g., $W_{i,n}$ is empty for all $i<\Omega((n-k)/\log n)$. Using these properties we show that the
resulting code has low ML decoding error.
(See Section~\ref{sec:inner}.)
\bnote{Add high level of the algorithm? mention the constructive version of weight distribution.}

\subsubsection{Making the code efficient}\label{sec:eff}
The above approach gives rise to a deterministic rateless code which achieves the
capacity of the BSC with a sub-exponential error of $\eps=2^{-\Omega(\beta/\log \beta)}$ where $\beta$ is the length of the information word.
However, the time complexity of
encoding/decoding the $n$-bit prefix of a codeword is $n\cdot 2^{O(\beta)}$.
We solve this problem by noting that Forney's concatenation technique~\cite{Forney66} naturally extends to the rateless setting.
We sketch the construction below. (Full details appear in Section~\ref{sec:concat}.)

The construction uses the inefficient rateless code as an ``inner code'' $\Cin:\bit^\beta \ra \bit^*$, and, in addition, employs a standard efficient outer code $\Cout:B^{\kout}\ra B^{\nout}$ where $B\triangleq\BIT{\beta}$ and $\kout\triangleq k/\beta$.

To encode a message $m\in \BIT{k}$, we parse it as $M\in B^{\kout}$,
apply the outer code to obtain a codeword $C\triangleq(C_{1},\ldots,C_{\nout})$ and then apply
the inner code to each of the symbols of $C$ in parallel. Namely, each symbol $C_i$ is
encoded by the code $\Cin$ to an infinitely-long column vector. The $\nin \cdot \nout$ prefix of the concatenated encoding is obtained by collecting the binary
vectors $(X_1,\ldots,X_{\nout})$ where $X_i$ denotes the prefix of length $\nin$ of
the inner codeword that corresponds to $C_i$.

Decoding proceeds in the natural way. Let $Y=(Y_1,\ldots,Y_{\nout})$ denote
the noisy $\nin\cdot \nout$ prefix of the encoding of the message $m$. First, maximum
likelihood decoding is employed to decode each of the inner codewords $Y_i$ into $\hat{X}_i$. Next, the decoder of the
outer code recovers an information word $M$ from the noisy codeword
$(\hat{X}_1,\ldots,\hat{X}_{\nout})$.

In order to prove Theorem~\ref{thm:main}, we need a somewhat non-standard setting of
the parameters. To avoid having to fix the gap to the channel's capacity ahead of
time, we use an outer code whose rate tends to $1$ (i.e.,
$\nout=\kout(1+o(1))$). Set $\beta=\omega(1)$. For concreteness, take an outer code $\Cout:B^{\kout}\ra
B^{\nout}$ with $\nout=\kout+\kout/\poly(\beta)$, and assume that the code can be decoded from a fraction of
$\eps'=\Omega(1/\poly(\beta))$ errors in time $\nout \cdot \poly(\beta)$ and can be
encoded with similar complexity.\footnote{Such a code can be obtained based on expander graphs, e.g.,~\cite{Spielman96a,spielmanPhD,guruswami2005linear}.
In fact, we will employ the code of~\cite{guruswami2005linear} which achieves a smaller alphabet of absolute size $\beta$. This is not a real issue as we can increase the alphabet to $2^{\beta}$ by parsing $\beta/\log \beta$ symbols as a single symbol without affecting the properties of the code. See Section~\ref{sec:concat}.} A standard application of Chernoff's bound shows that the decoding error of $p$-noisy codeword of length
$n \geq k\cdot\frac{1}{C(p)-\delta}$, is $2^{-\Omega(\nout(\eps'-\eps)^2)}$, which, under
our choice of parameters, simplifies to $2^{-\Omega(k/\poly(\beta))}$. For a slowly increasing
$\beta=\omega(1)$, we derive an almost-exponential error, and an almost linear
encoding/decoding time complexity of $\nout \cdot \beta+n \cdot 2^{O(\beta)}$.

Theorem~\ref{thm:weak} is obtained by using a (large) constant $\beta$ which depends on the gap to capacity $\delta$. As a result the rate of the outer code is bounded away from $1$, but the error becomes exponentially small and both encoding and decoding can be performed in linear time.

\subsection{Discussion}
One of the main conceptual contributions of this work is a formalization of rateless codes from an algorithmic point of view (see Section~\ref{sec:rateless}).
This formulation raises a more general research problem:
\begin{quote}
    Is it possible to gradually generate an infinite combinatorial object $\mathcal{O}=\set{\mathcal{O}_i}_{i=1}^\infty$ 
    via a deterministic algorithm?
\end{quote}
 Note that the question may be interesting even for inefficient algorithms as it may be infeasible, in general, to decide whether a finite sequence $\mathcal{O}_1,\ldots, \mathcal{O}_n$ is a prefix of some good infinite sequence $\mathcal{O}$. (This is very different than the standard finite setting, where inefficient derandomization is trivially achievable by exhaustive search.) It will be interesting to further explore other instances of this question (e.g., for some families of graphs).

 The formulation of a deterministic construction of a rateless code can be formulated as
 follows. Refer to a generating matrix as ``pseudo-random-weight'' if the weight
 distribution of the code it generates is ``close'' to the expected weight
 distribution of random linear codes.  Our main technical contribution is a
 deterministic construction of an infinite generating matrix, every finite prefix of
 which is ``pseudo-random-weight''.

 An interesting open problem is to obtain stronger approximations for the ``ideal''
weight distribution.  Specifically, it should be possible to improve the code's distance
 from sub-linear ($\Omega((n-k)/\log n)$) to linear ($\Omega((n-k))$) in the
 redundancy. More ambitiously, is it possible to construct a rateless code which, for
 every restriction to $n$ consecutive bits, achieves the capacity of the BSC? Getting
 back to our motivating story of noisy multicast, such a rateless code would allow
 the receivers to dynamically join the multicast.

\section{Rateless Codes}\label{sec:rateless}
In this section we formalize the notion of rateless codes. We begin with some standard notation.
\paragraph{Notation.}
The Hamming distance between two binary vectors $x,x'$ of equal length is denoted by
$\dist(x,x')$.  Let $\mu$ denote a probability distribution and $X$ denote a random
variable.  We denote that $X$ is distributed according to $\mu$ by $x\rs\mu$.  Let
$\BSC(p)$ denote the binary symmetric channel with crossover probability $p\in
(0,\half)$. We abuse notation and write $\noise\rs \BSC(p)$ to denote that $\noise$
is a binary vector whose coordinates are random independent Bernoulli trials chosen
to be $1$ with probability $p$ and a $0$ with probability $1-p$. (The vector's length
will be clear from the context.) Recall that the capacity of the binary symmetric
channel is $1-H(p)$ where $H(p)\triangleq -p \log p-(1-p)\log p$ is the entropy function. (By default, the base of all logarithms is $2$.)

We begin with a syntactic definition of a rateless code.
\begin{definition}[rateless code]
  A rateless code is a pair of algorithms $(\enc,\dec )$.
  \begin{enumerate}
  \item The encoder $\enc: \BIT{*} \times \N \rightarrow
    \bit$ takes an information word $m\in\BIT{*}$ and an index $i\in\N$, and outputs the $i$-th bit of the encoding of $m$. (Equivalently, the encoding of $m$ is an infinite sequence of bits $(\enc(m,i))_{i\in \N}$.)

  \item The decoder $\dec:\BIT{*} \times \N \rightarrow \BIT{*}$
   maps a noisy codeword $y\in\str$ and an integer $k$ (which corresponds to the length of the information word) to an information word $m'\in \bit^k$.
  \end{enumerate}
\end{definition}
Note that in our definition, both the encoder and the decoder are assumed to be \emph{deterministic}. One can relax the definition and consider a probabilistic rateless code in which the encoder and the decoder depend on some shared randomness. This corresponds to an ensemble of codes from which a code is randomly chosen.

\paragraph{Conventions.}
We let $\enc(m,[1:n])$ denote the first $n$ bits of the codeword that corresponds
to $m\in\BIT{*}$.  Namely, $\enc(m,[1:n])$ is the binary string $c=(c_1,\ldots,c_{n})$, where
$c_i=\enc(m,i)$. A rateless code defines $(n,k)$ codes for every $n$ and $k$ via \short{ $C_{n,k}\triangleq \{ \enc(m,[1:n]) \mid m\in \BIT{k}\}$.}
 \full{ \begin{align*}
    C_{n,k}\triangleq \{ \enc(m,[1:n]) \mid m\in \BIT{k}\}
\:.
\end{align*}} We measure the complexity of encoding (resp. decoding) of a rateless
code as the time $T(k,n)$ that takes to encode (resp., decode) the code $C_{n,k}$.
The encoder and the decoder are defined for every information block length $k$.
We often consider a specific $k$ and then abbreviate $\dec(y,k)$ by $\dec(y)$.

\begin{remark}[Additional features.]
In some scenarios it is beneficial to have a rateless code with the following additional features.
\begin{itemize}
  \item (Linearity) A rateless code is \emph{linear} if $\enc$ is a linear function. Namely, for $m\in\GF(2)^k$, we have
  \full{\[\enc(m,i)= R_i \cdot m,\]} \short{$\enc(m,i)= R_i \cdot m$,} where $\{R_i\}_{i=1}^{\infty}$, is an infinite sequence of row vectors $R_i\in\GF(2)^k$. We refer to the infinite matrix $G=\{R_i\}_{i=1}^{\infty}$ as the \emph{generator} matrix of the code.

  \item (Systematic) An encoding is \emph{systematic} if, for every $m\in \BIT{k}$, we have
  $\enc(m,[1:k])=m$.
\end{itemize}
\end{remark}

We define the \emph{error function} of a rateless code $(\enc,\dec )$ over the binary
symmetric channel $\BSC(p)$ as a function of $k,n$ and $p\in (0,1/2)$. \short{We
  write $\noise\rs \BSC(p)$ to denote a vector of independent Bernoulli trials that
  attain the value $1$ with probability $p$ and $0$ with probability $1-p$.}
\begin{definition}[The error function]
  \[\er(p,k,n) \triangleq \max_{m\in \bit^k}\Pr_{\noise\rs
    \BSC(p)}[\dec(\enc(m,[1:n])+\noise)\neq m].\]
\end{definition}
Equivalently, this is the maximum error probability, over the $\BSC(p)$, of the code
$C_{n,k}$ that is obtained by restricting the rateless code to a prefix of
length $n$.

\begin{definition}[capacity achieving rateless code for $\BSC$]
 A rateless code $( \enc,\dec)$ achieves capacity with respect to the binary
 symmetric channel if, for every $p\in (0,1/2)$ and every $\delta \in (0,1-H(p))$,
if
$n(k)\triangleq \frac{k}{1-H(p)-\delta}$, then

 \begin{equation}\label{Eq:err}
        \lim_{k \to \infty} \er(p,k,n(k)) = 0.
\end{equation}
\end{definition}
\noindent Naturally, it is desirable to bound~(\ref{Eq:err}) by a quickly decaying
function of $k$.

Motivated by the analysis of finite codes, one may be interested also
in proving that, for a fixed $k$, increasing redundancy over the same channel also
increases the probability of successful decoding, namely
\begin{align*}
\forall ~k\qquad   \lim_{n \to \infty} \er(p,k,n) = 0.
\end{align*}
Such a property implies that the minimum distance increases as a function of $n$ and
that the decoding algorithm benefits from this increase.
\section{An Inefficient Deterministic Rateless Code}
\label{sec:inner}
In this section we present an (inefficient) deterministic construction of a rateless
code that achieves capacity with respect to binary symmetric channels. In fact, when
all other parameters are fixed, the error function decreases almost exponentially as
a function of $n$. This code will be later used as the inner code of our final
construction.  Formally, we prove the following theorem.

\begin{theorem}\label{thm:inner}
  There exists a deterministic, rateless, linear, systematic code $(\enc,\dec)$ with the following
  properties:
  \begin{description}
  \item[Capacity achieving:] For every $p\in (0,\half)$ and  $\delta\in (0,1-H(p))$, if $n\geq k/(1-H(p)-\delta)$, then
    the error function satisfies\footnote{Note that if $n=k/(1-H(p)-\delta)$, then
      the theorem simply states that the error function is $e^{-O(k/\log k)}$.
However, the bound also holds for rates far below the capacity. For example, if $k$
is constant and $n$ tends to infinity, then the error function is $e^{-\Omega(n/\log n)}$.}
    \[\er(p,k,n)=e^{-\Omega(n/\log n)}\:.\]
  \item[Complexity:] Encoding and decoding of $k$-bit information words and
    $n$-bit codewords can be done in time $O(nk\cdot 2^{2k})$.
  \end{description}
\end{theorem}
The decoder is simply maximum likelihood decoding. The encoder multiplies the
information word by the generating matrix.
Each row of the generating matrix can be computed in
time $O(k\cdot 2^{2k})$. Hence, the generating matrix of $C_{n,k}$ can be computed in
time $O(nk\cdot 2^{2k})$. Both the encoder and decoder require the generating matrix.
Once the generating matrix of $C_{n,k}$ is computed, the running times of the encoding and the decoding
are as follows:
    \begin{itemize}
    \item The encoding of $\enc(m,[n:1])$ of $m\in \BIT{k}$ can be computed in time
      $O( n\cdot k)$.
    \item Computing $\dec(y,k)$ for $y\in\BIT{n}$ can be done in $O(n\cdot k \cdot
      2^k)$.
    \end{itemize}

In the following sections we describe the construction of the generating matrix of
the code and analyze the error of the maximum likelihood decoder.
\subsection{Computing the generating matrix}

Our goal is to construct an infinite generating matrix $G$ with $k$ columns.  Let
$R_i\in \bit^k$ denote the $i$th row of the generating matrix.  Let $G_n$ denote the
$k\times n$ matrix, the rows of which are $(R_i)_{i=1\ldots n}$. Let $C_{n,k}$ denote
the code generated by $G_n$.  The generating matrix $G$ begins with the $k\times k$
identity matrix, and hence each code $C_{n,k}$ is systematic.  Subsequent rows $R_i$
(for $i>k$) of the generating matrix are constructed one by one.  Let
$W_{i,n}\triangleq \{x\in\bit^k : \wt(G_n\cdot x)=i\}$ denote the $i$th weight class
of $C_{n,k}$. The rows are chosen so that the weight distribution $(|W_{1,n}|,\ldots,
|W_{n,n}|)$ of $C_{n,k}$ is close to that of a random $[n,k]$-linear code
$C^*_{n,k}$.  Note that when a row vector $R_{n+1}$ is added, if $x\in\bit^k$ is
orthogonal to $R_{n+1}$, then $\wt(G_{n+1}\cdot x)=\wt(G_{n}\cdot x)$; otherwise,
$\wt(G_{n+1}\cdot x)=\wt(G_{n}\cdot x)+1$.  Thus $R_{n+1}$ splits each weight class
$W_{i,n}$ to two parts: the orthogonal vectors which ``remain'' in $W_{i,n+1}$, and
the non-orthogonal vectors which are ``elevated'' to $W_{i+1,n+1}$.

\begin{definition}
A vector $R\in \GF(2)^k$ $\eps$-\emph{splits}
 a set $S\subseteq \GF(2)^k$ if
\[
(\frac 12 - \eps)\cdot |S| \leq |\{s\in S \mid s\cdot R = 1\}| \leq (\frac 12 + \eps)\cdot |S|.
\]
A vector $R\in \GF(2)^k$ $\eps$-\emph{elevates}
 a set $S\subseteq \GF(2)^k$ if
\[
|\{s\in S \mid s\cdot R = 1\}| \geq \eps \cdot |S|.
\]
\end{definition}

Ideally, we would like to find a row $R_{n+1}$ that $\eps$-splits every weight class
$W_{i,n}$. Since we cannot achieve this, we compromise on splitting only part of the
weight classes, as follows.
By a probabilistic argument, there exists a single vector which $\eps$-splits all
weight classes that are large (where a weight class $W_{i,n}$ is large if $|W_{i,n}|\geq
2n^2$). However, we cannot find vector that also $\epsilon$-splits every weight class
that is small.

The algorithm for computing the rows $R_i$ of $G$ for $i>k$ is listed as
Algorithm~\ref{alg:Rn}.  The algorithm employs a marking strategy to deal with small
weight classes $W_{i,n}$.  Initially, all the nonzero information words are unmarked.
Once an information word becomes a member of a small weight class, it is marked, and
remains marked forever (even if it later belongs to a weight class $W_{i',n'}$ which is
large). The unmarked vectors in $W_{i,n}$ are denoted by $\widehat{W}_{i,n}$. By
definition, the set $\widehat{W}_{i,n}$ is either empty or large, and so there exists
a vector $R_{n+1}$ which $\eps$-splits $\widehat{W}_{i,n}$. In addition, $R_{n+1}$ is
required to elevate the set of nonzero codewords of minimum weight.  As we will
later see, the distance of the resulting code grows sufficiently fast as a function
of $n$, and its weight distribution is sufficiently close to the expected weight
distribution of a random linear
code. 

\begin{algorithm}
  \caption{Compute-Generating-Matrix - An algorithm for computing rows $R_n$ of the
    generating matrix of the rateless code for $n>k$. }
\begin{enumerate}
\item Let $(R_1,\ldots,R_k)$ be the rows of the $k\times k$ identity matrix.
\item Initialize the set of marked information words $M\gets \emptyset$.
\item For $n=k$ to $\infty$ do
  \begin{enumerate}
\item For $1\leq i\leq n$, let $W_{i,n}$ be the set of information words that are
  encoded by a codeword of weight $i$.
\item Let $d>0$ be the minimal positive integer for which $W_{d,n}$ is non-empty.
\item For every $i$, if $|W_{i,n}\setminus M|<2n^2$, then mark all the information
  words in $W_{i,n}$ by setting
  $M\gets M \cup W_{i,n}$. Let $\widehat{W}_{i,n}\triangleq (W_{i,n}\setminus M)$
  denote the unmarked vectors in $W_{i,n}$.
\item\label{line:Rn} Let $R_{n+1}$ be the lexicographically first vector in
  $\GF(2)^k$ that simultaneously $\frac{1}{2\sqrt{n}}$-splits every unmarked weight
  class $\widehat{W}_{i,n}$ and $1/8$-elevates $W_{d,n}$.
  \end{enumerate}
\end{enumerate}
\label{alg:Rn}
\end{algorithm}

We remark that (according to the analysis) the $1/8$-elevation of $W_{d,n}$ can be skipped if
$\widehat{W}_{d,n}\neq\emptyset$ (namely, the elevation is required only if every
vector in $W_{d,n}$ is marked).
It is not hard to verify that Algorithm~\ref{alg:Rn} can compute the first $n$ rows
in time $O(nk\cdot 2^{2k})$.  The following lemma states that Algorithm~\ref{alg:Rn}
succeeds in finding a row $R_n$ for every $n>k$.
\begin{lemma}\label{lem:Rn}
The algorithm always finds a suitable vector $R_{n+1}$ in Line~\ref{line:Rn}.
\end{lemma}
The lemma is proven via a simple probabilistic argument. See Appendix~\ref{app:Rn}.
\subsection{Weight Distribution}
In this section we analyze the weight distribution of the linear code $C_{n,k}$.
We let $w_{i,n}$ be the size
of $W_{i,n}$, the set of information words whose encoding under $C_{n,k}$ has Hamming
weight $i$. We will show that $w_{i,n}$ is not far from the expected weight
distribution $w^*_i(n,k) \triangleq {n \choose i} \cdot 2^{k-n}$ of a random $[n,k]$
linear code. \short{Let $\Pikn\triangleq \prod_{j=k+1}^{n-1}
  \left(1+\frac{1}{\sqrt{j}}\right)=O(e^{2(\sqrt{n}-\sqrt{k})}).$}

\begin{observation}\label{obs:M}
After $n$ iterations, the number of marked information words is less than $2n^4$.
\end{observation}
\begin{proof}
  For every $i,n'\leq n$ the set $W_{i,n'}$ contributes less than $2n^2$ information
  words to the set $M$ of marked words. Hence there are most $2n^4$ marked vectors
  after the $R_n$ is chosen.
\end{proof}

\begin{myclaim}\label{claim:weightRLC}
  For every $n$ and $i$, we have that $w_{i,n}\leq 2n^4+w^*_i(n,k) \cdot \Pikn$\short{.} \full{where
  \[
\Pikn\triangleq \prod_{j=k+1}^{n-1} \left(1+\frac{1}{\sqrt{j}}\right)\leq e^{2(\sqrt{n}-\sqrt{k})}.
\]}
\end{myclaim}
\begin{proof}
By Observation~\ref{obs:M}, it suffices to bound the unmarked vectors by
\begin{align}\label{eq:win}
  |\widehat{W}_{i,n}| \leq w^*_i(n,k) \cdot \Pikn\:.
\end{align}
Indeed, $|\widehat{W}_{i,n}|$ and $w^*_i(n,k)$ satisfy the following
recurrences:
{\small
\begin{align*}
w^*_i(n,k)&=\frac 12 \cdot (w^*_{i-1,n-1} + w^*_{i,n-1})\\
|\widehat{W}_{i,n}| &\leq
\left(1+\frac{1}{\sqrt{n-1}}\right)\cdot \frac 12
\cdot \left(|\widehat{W}_{i-1,n-1}| +|\widehat{W}_{i,n-1}|\right)\:.
\end{align*}
}
We can now prove Eq.~\ref{eq:win} by induction on $n\geq k$. Indeed,
$w_{i,k}=w_{i,k}^*$, and
\begin{align*}
|\widehat{W}_{i,n}| &\leq
\left(1+\frac{1}{\sqrt{n-1}}\right)\cdot \frac 12
\cdot \left(w^*_{i-1,n-1}\Pi_{k,{n-1}} +w^*_{i,n-1}\Pi_{k,{n-1}}\right)\\
&=
 \frac {1}{2} \left(w^*_{i-1,n-1} +w^*_{i,n-1} \right)\Pikn\\
 &= w^*_i(n,k) \Pi_{k,n}.
\end{align*}
The claim follows.
\myqed\end{proof}

\noindent We will also need to prove that the distance of $C_{n,k}$ is sufficiently large.
\begin{myclaim}\label{claim:dist}
For every $n>k$, the minimum distance of the code $C_{n,k}$ is greater than
$\frac{n-k}{55\cdot \log n}$.
\end{myclaim}
\begin{proof}
 It is easier to view the evolution of the weight distribution of $C_{n,k}$ as a
  process of shifting balls in $n$ bins.  A ball represents a nonzero information
  word, and a bin corresponds to a weight class.  We assume that $\bin(1)$ is positioned on the
  left, and $bin(n)$ is positioned on the right.  Moving (or shifting) a ball one bin
  to the right means that the augmentation of the generating matrix by a new row
  increases the weight of the encoding of the information word by one. Note that, as
  the generating matrix is augmented by a new row, a ball either stays in the same
  bin or is shifted by one bin to the right.

  Step $t$ of the process corresponds to the weight distribution of $C_{n',k}$ for
  $n'=t+k$.  Let $\bin_t(i)$ denote the set of balls in $\bin(i)$ after step
  $t$.
By Algorithm~\ref{alg:Rn}, the process treats marked balls and unmarked balls differently.

  Let $t\triangleq(n-k)/2$ denote half the redundancy.
Let $\alpha\triangleq \frac{2}{\log_2 (8/7)} < 11$.
Let $\Delta\triangleq\frac{n-k}{\alpha \log (2n^4)}$.
In these terms, We prove a slightly stronger minimum distance, namely,
\begin{equation}
  \label{eq:goal}
  \bin_{2t}(i)=\emptyset, \qquad \forall i \leq \Delta.
\end{equation}

  The proof is divided into two parts. First we consider the unmarked balls, and then
  we consider the marked balls.  We begin by proving that
\begin{equation}\label{eq:unmarked}
    bin_t(i)\setminus M= \emptyset, \qquad \forall i\leq \Delta.
\end{equation}
Namely, after $t$ iterations of Algorithm~\ref{alg:Rn}, the bins
$\bin(1),\ldots,\bin(\Delta)$ may contain only marked balls.  Note that
if $\bin_t(i)=\emptyset$ for every $i\leq \Delta$, then
$\bin_{2t}(i)=\emptyset$ for every $i\leq \Delta$.

The proof of Equation~\ref{eq:unmarked} is based Claim~\ref{claim:f} (proved in
Appendix~\ref{appendix:bound}) that states the following:
\begin{equation}
|\bin_{t}(i)| \leq \left(\frac 23\right)^{t} \cdot \binom{k+t}{i} \leq \left(\frac 23\right)^{t} \cdot (k+t)^i.
\label{eq:W}
\end{equation}
The intuition is as follows. Initially, $\bin_0(i)$ contains at most
$\binom{k}{i}$ vectors.  After step $t+1$, $\bin_{t+1}(i)$ contains roughly
half the balls of $\bin_t(i-1)$ (i.e., the elevated balls) and roughly half
the balls of $\bin_t(i)$ (i.e. the non-elevated balls). A recursive
analysis shows that after $t$ steps we get the above expression (for simplicity the
bound assumes only $1/3$-elevation) .

For $t=(n-k)/2$ and $i\leq \Delta$, the RHS of Eq.~\ref{eq:W} is smaller than 1, and
so Eq.~\ref{eq:unmarked} follows.

To prove that $\bin_{2t}(i)\cap M = \emptyset$ for every $i\leq \Delta$, let
$t(i)\triangleq t+i\cdot \log_{8/7} (2n^4)$.  Note that $t(\Delta)= 2t$.
We wish to prove, by induction on $i$, that the leftmost bin with a marked ball after $t(i)$ iterations is
$\bin({i+1})$. After $\log_{8/7} (2n^4)$ additional iterations, also $\bin(i+1)$
lacks marked balls. In this manner, after $2t$ iterations all the marked balls are
pushed to the right of $\bin({\Delta})$. Formally,
we claim that
\begin{equation}
  \label{eq:marked}
\bin_{t(i)}(j) \cap M = \emptyset, \qquad \forall j\leq i.
\end{equation}
Equation~\ref{eq:marked} suffices because $t(\Delta)=2t$, and hence it implies that
$\bin_{2t}(j)=\emptyset$ for every $j\leq \Delta$, as required.  The proof of
Eq.~\ref{eq:marked} is by induction on $i$. For $i=0$ the claim is trivial (because
every nonzero information word is encoded to a nonzero word). The induction step for
$i>0$ is as follows. For every $t(i-1) < t \leq t(i)$, if $\bin_t(i)$ contains a
marked ball, then, by the induction hypothesis, it is the leftmost bin that contains
a marked ball. Hence, each new row $R_{t+1}$ of the generator matrix $1/8$-elevates
$\bin_t(i)$.  Since $\bin_t(i)$ consists only of marked balls, by Obs.~\ref{obs:M},
it follows that $|\bin_{t(i-1)}(i)|< 2n^4$.  Hence, after $\log_{8/7}(2n^4)$ steps,
the bin is emptied, namely, $\bin_{t(i)}(i)=\emptyset$, as required.

We proved that $bin_{2t}(i)$ is empty if $i\leq \Delta$, and the claim follows.\myqed\end{proof}
\full{
Overall Claims~\ref{claim:dist} and~\ref{claim:weightRLC} imply that $C_{n,k}$ is
close to an ``average'' code in the following sense.
Let $\alpha\triangleq \frac{2}{\log_2 (8/7)} < 11$.

\begin{lemma}\label{lemma:weight distribution}
  The weight distribution of the constructed code $C_{n,k}$ satisfies the following
  bound:
  \begin{align}
    \label{eq:weight}
    w_{i,n} &\leq
    \begin{cases}
      0 & \text{if $0<i\leq \frac{n-k}{\alpha\log (2n^4)}$}\\
2n^4+w^*_i(n,k) \cdot \Pikn & \text{if $i> \frac{n-k}{\alpha \log (2n^4)}$}.
    \end{cases}
  \end{align}
\end{lemma}
}

\subsection{Analysis of the ML Decoding Error}
In this section we complete the proof of Theorem~\ref{thm:inner}.  Let $\dec$ be the maximum-likelihood (ML)
decoder which, given a noisy codeword $y\in \nbit$ and $k$, finds a closest codeword
$\hat{y}\in C_{n,k}$ and outputs the message $m\in \bit^k$ for which $G_n\cdot m = \hat{y}$.
\begin{lemma}\label{lemma:inner error}
For every $p$ and $\delta \in (0,1-H(p))$. If $n\geq\frac{k}{1-H(p)-\delta}$, then the
  error function of the maximum likelihood decoder satisfies
    \[\er(p,k,n)=e^{-\Omega(n/\log n)}\:.\]
\end{lemma}
\begin{proof}
  Fix $p$ and $\delta$, and consider $n$ and $k$ such that $n\geq\frac{k}{1-H(p)-\delta}$.  Let $\delta_{\GV}$ be
  the root $\delta\in(0,1/2)$ of the equation $H(\delta)=1-\frac kn$.  Since the code
  is linear, we may assume without loss of generality that the all zero codeword was
  transmitted.  Our goal is to upper-bound the event that $\hat{y}$, the codeword
  computed by the ML-decoder, is non-zero. We divide the analysis into two cases
  based on the Hamming weight of $\hat{y}$.

\paragraph{Case 1: $\hat{y}$ is of weight smaller than $\delta_{\GV}\cdot  n$.}
For a fixed codeword $y$ of weight $i>0$, erroneous decoding to $y$ corresponds to
the event that
the $\BSC(p)$ flipped at least $i/2$ bits in the support of $y$.
(The support of $y$ is the set $\{j: y_j=1\}$.)
This event happens with probability
\[\textstyle{P_i \triangleq \sum_{j=\lceil i/2 \rceil}^{i} \binom{i}{j}\cdot p^j\cdot (1-p)^{i-j}}.\]
By a union-bound, we can upper-bound the
probability of the event that $0<\wt(\hat{y})< \delta_{\GV}\cdot  n$
by
\begin{align}
\sum_{i=1}^{\delta_{GV}\cdot n-1} w_{i,n} \cdot P_i
&\leq  \sum_{i=(n-k)/(55\log n)}^{\delta_{GV}\cdot n-1} (2n^4+e^{2\sqrt{n}})  \cdot P_i, \label{eq:case1}
\end{align}
where the upper-bound $w_{i,n}\leq (2n^4+e^{2\sqrt{n}})$
follows from \full{Lemma~\ref{lemma:weight distribution}}\short{Claims~\ref{claim:dist} and~\ref{claim:weightRLC}} and from the fact
that $w^*_i(n,k)<1$ if $i/n< \delta_{\GV}$.  Below, we show that
\begin{equation}\label{eq:Pi}
  P_i\leq 2^{-\beta\cdot i}
\end{equation}
where $\beta \triangleq -\frac 12 \cdot \log_2 (4p(1-p))$ is positive since $p\in (0,\half)$.
It follows that the error probability~(\ref{eq:case1}) is upper-bounded by
\[
(2n^4+e^{2\sqrt{n}}) \cdot \sum_{i=(n-k)/(55\log n)}^{\delta_{GV}\cdot n} 2^{-\beta\cdot i}\leq e^{-\Omega(n/\log n)}.\]
It is left to prove Eq.~(\ref{eq:Pi}). Indeed, by definition, $P_i$ satisfies
\[
P_i \triangleq \textstyle{\sum_{j=\lceil i/2 \rceil}^{i} \binom{i}{j}\cdot p^j\cdot
     (1-p)^{i-j} \leq
 p^{i/2}\cdot(1-p)^{i/2} \cdot \sum_{j=\lceil i/2 \rceil}^{i} \binom{i}{j} \leq
 p^{i/2}\cdot(1-p)^{i/2} \cdot 2^i}, \]
which can be written as $(4p(1-p))^{i/2}$. Because  $p<1/2$, it follows that $\beta>0$, and $P_i\leq 2^{-\beta\cdot i}$, as required.

\paragraph{Case 2: $\hat{y}$ is of weight larger than $ \delta_{\GV}\cdot  n$.}
In this regime, the spectrum of our code is sufficiently close to that of a random
linear code, and so the error of the ML-decoding can be analyzed via (an extension of) Poltyrev's bound~\cite{poltyrev1994bounds} (see also~\cite{shulman1999random}).
The extension bounds the probability of the event that ML-decoding returns a
``heavy'' word. Note that no assumption is made on the minimum distance of the code.
The proof is based on an analysis in~\cite{Barg_enee739lecture4}.
\begin{theorem}[extension of Thm.~1 of~\cite{poltyrev1994bounds} - proof in Appendix~\ref{appendix:P}]\label{thm:poltyrev}
   Let $p\in (0,\half)$ be a constant, $\delta >0$ be a constant such that $\frac{k}{n}<1-H(p)-\delta$, and $\tau\in [0,1]$ be a threshold parameter. There exists a constant $\alpha>0$ for which the following holds. If $C$ is an $[n,k]$ linear code whose weight distribution $\{w_{i}(C_n)\}_i$
     satisfies \[w_i \leq      2^{(\delta / 3) n} \cdot w^*_i(n,k) \qquad \text{for every } i\geq \tau n.\]
  Then, the probability over $\BSC(p)$ that the all zero word is ML-decoded to a codeword of weight at least $\tau n$ is $2^{-\alpha n}$.
\end{theorem}

Since the weight distribution of our code satisfies the Poltyrev's criteria
for codewords of weight at least $\delta_{\GV}\cdot  n$, we
conclude that the decoding error in case (2) is $2^{-\Omega(n)}$.

By combining the two cases, we conclude that the error-probability is at most
$2^{-\Omega(n/\log n)}$, as
required.
\myqed\end{proof}


\section{Efficient Rateless Codes}\label{sec:concat}
In this section we will prove our main theorems and construct an efficient rateless
code $(\enc,\dec)$ that achieves the capacity of the binary symmetric channel.  We
define $(\enc,\dec)$ via its restriction $C_{n,k}$ to information words of length $k$
and codewords of length $n$.  Following the outline sketched in
Section~\ref{sec:eff}, we let $C_{n,k}$ be the concatenation of an $[\nout,\kout]$
outer code $\Cout$ and an $[\nin,\kin]$ inner code $\Cin$ defined as follows.

\paragraph{Inner Code.}
The inner code $\Cin$ is the inefficient rateless code described in
Section~\ref{sec:inner} restricted to input length $\kin$ and output
length $\nin$. Recall that this is an $[\nin,\kin]$ linear systematic
code over $\{0,1\}$ which can be encoded in time $O(\nin\kin\cdot
2^{2\kin})$.  Maximum likelihood decoding requires $O(\nin\kin\cdot
2^{\kin})$ time and achieves an error of
$\er(p,\kin,\nin)=e^{-\Omega(\nin/\log \nin)}$ over $\BSC(p)$ as long
as $\nin\geq \kin \cdot (1-H(p)-\delta)^{-1}$ for some
$\delta\in (0,1-H(p))$. Both encoding and decoding can be implemented in
parallel time of $O(\kin)$.

\paragraph{Outer Code.}
The outer code $\Cout$ is taken from~\cite[Lemma 1]{guruswami2005linear}.  It is an
$[\nout,\kout]$ linear systematic code over an alphabet ${\Sigma_\out}$ with $\nout =
\kout \cdot (1+|\Sigma_{\out}|^{-1/2})$.  Hence, the rate of the outer code tends to one as the alphabet $\Sigma_\out$
increases.  The outer code can be encoded in time $O(\nout\cdot
|\Sigma_\out|^{1/2})$. Decoding in time $O(\nout\cdot |\Sigma_\out|)$ is successful
as long as the fraction of errors is bounded by
$\eps_\out=\Theta(|\Sigma_\out|^{-1})$.  Furthermore, the code can be encoded and
decoded in parallel time of $O(\log (\nout \cdot |\Sigma_{\out}| ))$.

\begin{construction}[The concatenated code $C^{\beta}_{n,k}$]
For lengths $k$ and $n$, and a parameter $\beta$ let
\begin{equation*}
    |\Sigma_\out|=\kin=\beta, \quad  \kout= k/\log_2 |\Sigma_\out|, \quad \Lin = (\nout \cdot \log_2 |\Sigma_\out|)/\kin, \quad \nin=n/\Lin.
\end{equation*}
\begin{itemize}
  \item The \emph{encoder} of the concatenated code $C^{\beta}_{n,k}$ maps $k$-bit information word to $n$-bit codeword as follows (see Figure~\ref{figure:enc}).
\[
F^k_2 \overset{1}{\hookrightarrow} \Sigma_\out^{\kout} \overset{2}{\longrightarrow} \Sigma_{\out}^{\nout}
\overset{3}{\hookrightarrow} (F_2^{\kin})^{\Lin} \overset{4}{\longrightarrow} (F_2^{\nin})^{\Lin}.
\]
The four steps of the encoder are:
(1)~A message $m\in \BIT{k}$ is parsed as the message $m_\out \in
(\Sigma_\out)^{\kout}$.
Namely, $\Sigma_\out=\{0,1\}^{\log \beta}$, and the message $m$ is broken into
$\kout$ blocks of length $\log_2 |\Sigma_\out|$.
(2)~The encoder of the outer code maps $m_\out$ to a codeword $c_\out \in  (\Sigma_\out)^{\nout}$.
(3)~The outer codeword $c_\out$ is parsed as $\Lin$  messages
$(m_{\inn}^{1},\ldots,m_{\inn}^{\Lin})$ each over $\BIT{\kin}$.
(4)~The encoder of the inner code maps each message $m_{\inn}^{j}$ to an inner
codeword $c_\inn^j\in \BIT{\nin}$.

 \item The \emph{decoder} of the concatenated code $C^{\beta}_{n,k}$ maps $n$-bit codeword word to $k$-bit information as follows (see Figure~\ref{figure:dec}).
\[
(F_2^{\nin})^{\Lin} \overset{4}{\longrightarrow}  (F_2^{\kin})^{\Lin} \overset{3}{\hookrightarrow} \Sigma_{\out}^{\nout}
\overset{2}{\longrightarrow}  \Sigma_\out^{\kout} \overset{1}{\hookrightarrow} F^k_2 .
\]
The four steps of the decoder correspond to the encoding steps in reveresed order:
(4)~The decoder of the inner code applies maximum likelihood decoding to each inner
noisy codeword $\hat{c}_\inn^j\in \BIT{\nin}$. We denote the ML-decoding of
$\hat{c}_\inn^j\in \BIT{\nin}$ by
$\hat{m}_{\inn}^{j}$.
(3)~ The $\Lin$  (inner) information words
$(\hat{m}_{\inn}^{1},\ldots,\hat{m}_{\inn}^{\Lin})$ each over $\BIT{\kin}$ are parsed
as a noisy codeword $\hat{c}_\out \in  (\Sigma_\out)^{\nout}$ of the outer code.
(2)~The decoder of the outer code maps the noisy codeword $\hat{c}_\out \in  (\Sigma_\out)^{\nout}$ to a message $\hat{m}_\out \in
(\Sigma_\out)^{\kout}$.
(1)~The message $\hat{m}_\out$ is parsed as a message $\hat{m}\in \BIT{k}$.

\end{itemize}
\end{construction}

The encoder of the rateless code (when $n$ is not predetermined) outputs the
encoding of $m_\inn^1,\ldots,m_\inn^{\Lin}$ ``row by row''. Namely, after the $i$'th
bit of the encodings is output, the encoder outputs bit $i+1$ of each inner-codeword.
Hence,  the code
$C^{\beta}_{n,k}$ is a prefix of the code $C^{\beta}_{n',k}$ for
$n<n'$ and so the code defines a rateless code. Also note that the
code is systematic and the complexity of encoding is $O(\nout\cdot
|\Sigma_\out|^{1/2}+ \Lin\cdot \nin\cdot \kin\cdot 2^{2\kin}
)=O(n\cdot \beta \cdot 2^{2\beta})$ and the complexity of decoding is
$O(\nout\cdot |\Sigma_\out|+ \Lin\cdot \nin\cdot \kin\cdot 2^{2\kin}
)=O(n\cdot \beta \cdot 2^{2\beta})$. (We assume that the encoder and the decoder
need to compute the generating matrix.) Furthermore, both operations can
be performed in parallel-time of $O(\kin + \log
(\nout \cdot |\Sigma_{\out}| ))=O(\beta +\log n)$.  The performance
over $\BSC(p)$ is analyzed by the following claim.

In the following claim we bound the decoding error of the concatenated code $C_{n,k}$
over $\BSC(p)$. We consider two settings. In the first setting, the rate of the inner
code is $(1-H(p)-\delta)$, and we prove that the probability of erroneous decoding
tends to zero almost exponentially in $k$. In the second setting, the outer code is
fixed (hence $k, \beta, \kout$, and $\nout$ are fixed), and the rate of the inner
code tends to zero. In the second setting we prove that the probability of erroneous
decoding tends exponentially to zero as a function of $n$. This implies that the
decoder benefits from the increase in the minimum distance of the code as $n$ increases.
\begin{myclaim}\label{claim:conc}
For every $p\in (0,\half)$ and $\delta>0$, if $\nin \geq \kin
\cdot\frac{1}{1-H(p)-\delta}$, then the decoding error $\er(p,k,n)$ of the concatenated code
$C_{n,k}$ over $\BSC(p)$ is $2^{-\Omega(\frac{k}{\beta^3})}$.
Moreover, if $p,k$, and the outer code are fixed, then
$\er(p,k,n)=2^{-\Omega(n/\log n)}$. 
\end{myclaim}
\begin{proof}
  Let $\hat{c}_\inn=(\hat{c}_\inn^1,\ldots,\hat{c}_\inn^{\Lin})$ denote the noisy
  prefix of length $n=\nin\cdot \Lin$ of the encoding of the message $m$.  Let
  $\hat{e}$ denote the fraction of the inner-code information words that are
  incorrectly decoded by the ML-decoder.  The decoder of the outer-code is successful
  as long as $\hat{e}<\eps_\out$. (Note that each decoded inner information word is
  parsed into $k_\inn/\log_2 |\Sigma_\out|$ symbols of the outer code. Hence, the
  fraction of erroneous symbols is bounded by $\hat{e}$.)  When $k$ tends to infinity,
  we bound the probability of the event that $\hat{e}\geq \eps_\out$ using an
  additive Chernoff bound.  Let $\eps_\inn$ denote the probability of erroneous
  decoding of a noisy inner codeword $\hat{c}_\inn^j$. As the ML-decoding errors are
  $\Lin$ independent random events, we conclude that $\Pr[\hat{e}\geq \eps_\out] \leq
  2^{-2\Lin(\eps_{\out}-\eps_\inn)^2}$.

  By Lemma~\ref{lemma:inner error}, $\eps_{\inn}=e^{-\Omega(\nin/\log
    \nin)}=e^{-\Omega(\beta/\log \beta)}$. Under our choice of parameters
  $\eps_{\out}-\eps_\inn=\Omega(1/\beta)$ and $\Lin=(\nout\cdot \log \beta)/\beta >(\kout\cdot \log \beta)/\beta=k/\beta$, and so the
 the bound on the error probability simplifies to $2^{-\Omega(k/\beta^3)}$.

 In the second setting, when the outer code is fixed, we bound the probability of the
 event that $\hat{e}\geq \eps_\out$ by a union bound over all $\eps_\out$-fractions
 of $\Lin$. Namely, $\Pr(\hat{e}\geq \eps_\out) \leq \binom{\Lin}{\eps_\out\cdot \Lin
 }\cdot {\eps_\inn}^{\eps_\out\cdot\Lin } $ which is bounded by $
 2^{H(\eps_\out)\cdot\Lin}\cdot \eps_\inn^{\eps_\out \cdot \Lin} $.  By
 Lemma~\ref{lemma:inner error}, $\eps_{\inn}=e^{-\Omega(\nin/\log \nin)}$.  Because
 $\eps_\out$ and $\Lin$ are fixed, the probability of the event is bounded by
 $e^{-\Omega(n/\log n)}$, as required.
 \myqed\end{proof}

Letting $\beta$ be an (arbitrary slowly) growing function of $k$ we derive the following corollary, which in turn, directly implies Theorem~\ref{thm:main}.
\begin{corollary}[Thm.~\ref{thm:main} refined]
  Let $\beta = \omega(1)$, the rateless code defined by $C^{\beta}_{n,k}$ is a linear
  systematic rateless code that can be encoded and decoded in time $O(n\cdot \beta
  \cdot 2^{2\beta}))$ and parallel time of $O(\log n + \beta)$. Furthermore, for
  fixed $\delta>0$ and crossover probability $p$ for which
  $n\geq k\cdot\frac{1}{1-H(p)-\delta}$, the decoding error is
  $2^{-\Omega(\frac{k}{\beta^3})}$.
\end{corollary}
\begin{proof}
  Since $\beta = \omega(1)$ the rate of the outer code is $1-o(1)$ and
  so for $n,p$ and $\delta$ which satisfy
  $\frac{n}{k} \geq \frac{1}{1-H(p)-\delta}$, we have
  that \[\frac{\nin}{\kin}\geq
  \frac{n}{k(1+\frac{1}{\sqrt{\beta}})}=\frac{1}{1-H(p)-\delta'}\] for
  $\delta'=\delta-o(1)$. We can therefore apply Claim~\ref{claim:conc}
  and derive the corollary.\myqed\end{proof}

The proof of Theorem~\ref{thm:weak} is similar, except that now, when we are given
the gap to capacity $\delta$ ahead of time, we can set $\beta$ to be a sufficiently large constant.

\begin{corollary}[Thm.~\ref{thm:weak} restated]
  Let $\delta>0$ be a constant. Then there exists a constant $\beta$
  for which the rateless code defined by $C^{\beta}_{n,k}$ is a linear
  systematic rateless code that can be encoded and decoded in time
  $O(n)$ and parallel time of $O(\log n)$. Furthermore, for crossover
  probability $p$ for which $n \geq k\cdot\frac{1}{1-H(p)-\delta}$, the
  decoding error is $2^{-\Omega(k)}$.
\end{corollary}
\begin{proof}
  Choose $\beta$ for which the rate of the outer code
  $\rout=\kout/\nout=1/(1+\delta/2)$. As a result, an $n$-bit prefix of the
  concatenated code of rate $R=k/n\geq1-H(p)-\delta$ implies that the rate of the
  inner code $\kin/\nin$ is at most $1-H(p)-\delta/2$, and so by
  Claim~\ref{claim:conc}, the decoding error $\er(p,k,n)\leq
  2^{-\Omega(k/\beta^3)}=2^{-\Omega(k)}$. By construction, encoding and decoding can
  be performed in linear-time and logarithmic parallel-time.  \myqed\end{proof}

\begin{figure}[htbp]
  \begin{center}
    \begin{tabular}[c]{p{0.75 \textwidth}}
\scalebox{1} 
{
\begin{pspicture}(0,-6.69)(16.062813,8.29)
\psframe[linewidth=0.04,dimen=outer](9.720938,6.49)(2.9009376,5.89)
\psframe[linewidth=0.04,dimen=outer](9.700937,4.33)(2.8809376,3.73)
\psline[linewidth=0.04cm](4.9209375,4.33)(4.9209375,3.77)
\psline[linewidth=0.04cm](5.9009376,4.33)(5.9009376,3.77)
\psline[linewidth=0.04cm](8.700937,4.31)(8.700937,3.75)
\psline[linewidth=0.04cm](3.9009376,4.31)(3.9009376,3.77)
\psline[linewidth=0.04cm,arrowsize=0.05291667cm 2.0,arrowlength=1.4,arrowinset=0.4]{<->}(2.9209375,4.51)(3.8809376,4.51)
\psline[linewidth=0.04cm,arrowsize=0.05291667cm 2.0,arrowlength=1.4,arrowinset=0.4]{->}(6.2809377,5.49)(6.2809377,4.67)
\psline[linewidth=0.04cm,arrowsize=0.05291667cm 2.0,arrowlength=1.4,arrowinset=0.4]{->}(6.3009377,3.29)(6.3009377,2.47)
\psframe[linewidth=0.04,dimen=outer](11.360937,2.11)(1.3209375,1.53)
\psline[linewidth=0.04cm](2.3209374,2.07)(2.3209374,1.55)
\psline[linewidth=0.04cm](10.320937,2.07)(10.320937,1.57)
\psline[linewidth=0.04cm](3.3009374,2.11)(3.3009374,1.55)
\psline[linewidth=0.04cm](4.3009377,2.11)(4.3009377,1.55)
\psdots[dotsize=0.08](6.7009373,3.91)
\psdots[dotsize=0.08](7.5009375,3.91)
\psdots[dotsize=0.08](7.1209373,3.91)
\psdots[dotsize=0.08](6.1009374,1.73)
\psdots[dotsize=0.08](6.5209374,1.73)
\psdots[dotsize=0.08](6.9009376,1.73)
\psline[linewidth=0.04cm,arrowsize=0.05291667cm 2.0,arrowlength=1.4,arrowinset=0.4]{<->}(1.3409375,2.29)(2.3009374,2.29)
\usefont{T1}{ptm}{m}{n}
\rput(3.4623437,4.8){$\log ( |\Sigma_\out| ) $}
\usefont{T1}{ptm}{m}{n}
\rput(1.9423437,2.58){$\log ( |\Sigma_\out| ) $}
\usefont{T1}{ptm}{m}{n}
\rput(6.952344,2.98){$(2)$}
\usefont{T1}{ptm}{m}{n}
\rput(6.8523436,5.02){$(1)$}
\psline[linewidth=0.04cm,arrowsize=0.05291667cm 2.0,arrowlength=1.4,arrowinset=0.4]{->}(6.3209376,-1.07)(6.3209376,-1.89)
\usefont{T1}{ptm}{m}{n}
\rput(3.7223437,-1.84){$1 $}
\usefont{T1}{ptm}{m}{n}
\rput(6.952344,-1.38){$(4)$}
\psline[linewidth=0.04cm](4.1009374,-2.29)(4.1009374,-5.47)
\psline[linewidth=0.04cm](4.7209377,-2.29)(4.7209377,-5.47)
\psline[linewidth=0.04cm](5.3209376,-2.31)(5.3209376,-5.49)
\psline[linewidth=0.04cm](8.480938,-2.29)(8.480938,-5.47)
\usefont{T1}{ptm}{m}{n}
\rput(3.7723436,-3.6){$c_\inn^1$}
\usefont{T1}{ptm}{m}{n}
\rput(8.832344,-3.6){$c_\inn^{\Lin}$}
\psline[linewidth=0.04cm,arrowsize=0.05291667cm 2.0,arrowlength=1.4,arrowinset=0.4]{<->}(3.4809375,-2.09)(4.0809374,-2.09)
\usefont{T1}{ptm}{m}{n}
\rput(8.882343,-1.84){$1 $}
\psline[linewidth=0.04cm,arrowsize=0.05291667cm 2.0,arrowlength=1.4,arrowinset=0.4]{<->}(8.480938,-2.09)(9.140938,-2.11)
\psdots[dotsize=0.08](7.5009375,-4.87)
\psdots[dotsize=0.08](7.1009374,-4.87)
\psdots[dotsize=0.08](6.7209377,-4.87)
\psline[linewidth=0.04cm](3.4809375,-2.29)(3.4809375,-5.47)
\psline[linewidth=0.04cm](9.120937,-2.27)(9.120937,-5.47)
\psline[linewidth=0.04cm](3.4609375,-2.29)(9.140938,-2.29)
\psdots[dotsize=0.08](3.7209375,-5.65)
\psdots[dotsize=0.08](3.7209375,-6.05)
\psdots[dotsize=0.08](3.7209375,-6.43)
\psdots[dotsize=0.08](8.920938,-5.65)
\psdots[dotsize=0.08](8.920938,-6.05)
\psdots[dotsize=0.08](8.920938,-6.43)
\usefont{T1}{ptm}{m}{n}
\rput(1.4123437,4.0){$m_\out$}
\usefont{T1}{ptm}{m}{n}
\rput(0.6123437,1.82){$c_\out$}
\psline[linewidth=0.04cm,arrowsize=0.05291667cm 2.0,arrowlength=1.4,arrowinset=0.4]{->}(6.3009377,1.11)(6.3009377,0.29)
\psframe[linewidth=0.04,dimen=outer](11.340938,-0.11)(1.3009375,-0.69)
\psline[linewidth=0.04cm](3.3409376,-0.11)(3.3409376,-0.67)
\psdots[dotsize=0.08](6.1409373,-0.45)
\psdots[dotsize=0.08](6.5609374,-0.45)
\psdots[dotsize=0.08](6.9409375,-0.45)
\psline[linewidth=0.04cm,arrowsize=0.05291667cm 2.0,arrowlength=1.4,arrowinset=0.4]{<->}(1.3209375,0.07)(3.3209374,0.07)
\usefont{T1}{ptm}{m}{n}
\rput(2.2023437,0.36){$\kin $}
\usefont{T1}{ptm}{m}{n}
\rput(6.952344,0.8){$(3)$}
\psline[linewidth=0.04cm](5.3209376,2.07)(5.3209376,1.55)
\psline[linewidth=0.04cm](5.2809377,-0.13)(5.2809377,-0.65)
\psline[linewidth=0.04cm](9.300938,-0.13)(9.300938,-0.65)
\psline[linewidth=0.04cm](9.320937,2.07)(9.320937,1.55)
\usefont{T1}{ptm}{m}{n}
\rput(2.3823438,-0.42){$m_{\inn}^{1} $}
\usefont{T1}{ptm}{m}{n}
\rput(10.382343,-0.42){$m_{\inn}^{\Lin} $}
\psline[linewidth=0.04cm,arrowsize=0.05291667cm 2.0,arrowlength=1.4,arrowinset=0.4]{<->}(9.300938,0.09)(11.300938,0.09)
\usefont{T1}{ptm}{m}{n}
\rput(10.1823435,0.38){$\kin $}
\psline[linewidth=0.04cm](3.1209376,6.45)(3.1209376,5.91)
\psline[linewidth=0.04cm](3.3209374,6.47)(3.3209374,5.93)
\psline[linewidth=0.04cm](3.5009375,6.45)(3.5009375,5.91)
\psline[linewidth=0.04cm](3.7009375,6.47)(3.7009375,5.93)
\psline[linewidth=0.04cm](3.9009376,6.45)(3.9009376,5.91)
\psline[linewidth=0.04cm](4.1009374,6.45)(4.1009374,5.91)
\psline[linewidth=0.04cm](4.3009377,6.45)(4.3009377,5.91)
\psline[linewidth=0.04cm](4.5009375,6.45)(4.5009375,5.91)
\psline[linewidth=0.04cm](4.7209377,6.45)(4.7209377,5.91)
\psline[linewidth=0.04cm](4.9209375,6.47)(4.9209375,5.93)
\psline[linewidth=0.04cm](5.1009374,6.45)(5.1009374,5.91)
\psline[linewidth=0.04cm](5.3009377,6.47)(5.3009377,5.93)
\psline[linewidth=0.04cm](5.5009375,6.45)(5.5009375,5.91)
\psline[linewidth=0.04cm](5.7009373,6.45)(5.7009373,5.91)
\psline[linewidth=0.04cm](5.9009376,6.45)(5.9009376,5.91)
\psline[linewidth=0.04cm](8.680938,6.45)(8.680938,5.91)
\psline[linewidth=0.04cm](8.900937,6.45)(8.900937,5.91)
\psline[linewidth=0.04cm](9.100938,6.45)(9.100938,5.91)
\psline[linewidth=0.04cm](9.300938,6.45)(9.300938,5.91)
\psline[linewidth=0.04cm](9.500937,6.45)(9.500937,5.91)
\usefont{T1}{ptm}{m}{n}
\rput(0.59234375,-0.42){$c_\out$}
\psdots[dotsize=0.08](6.7009373,6.11)
\psdots[dotsize=0.08](7.1209373,6.11)
\psdots[dotsize=0.08](7.5009375,6.11)
\usefont{T1}{ptm}{m}{n}
\rput(3.4423437,4.02){$m_{\out}^{1}$}
\usefont{T1}{ptm}{m}{n}
\rput(9.1823435,4.02){$m_{\out}^{\kout}$}
\usefont{T1}{ptm}{m}{n}
\rput(1.9023438,1.82){$c_{\out}^{1}$}
\usefont{T1}{ptm}{m}{n}
\rput(10.862344,1.8){$c_{\out}^{\nout}$}
\usefont{T1}{ptm}{m}{n}
\rput(1.4723438,6.16){$m\in \BIT{k}$}
\end{pspicture} 
}
    \end{tabular}
  \end{center}
\caption{Encoder: concatenation of the outer code and the inner code}
\label{figure:enc}
\label{fig:encoder}
\end{figure}
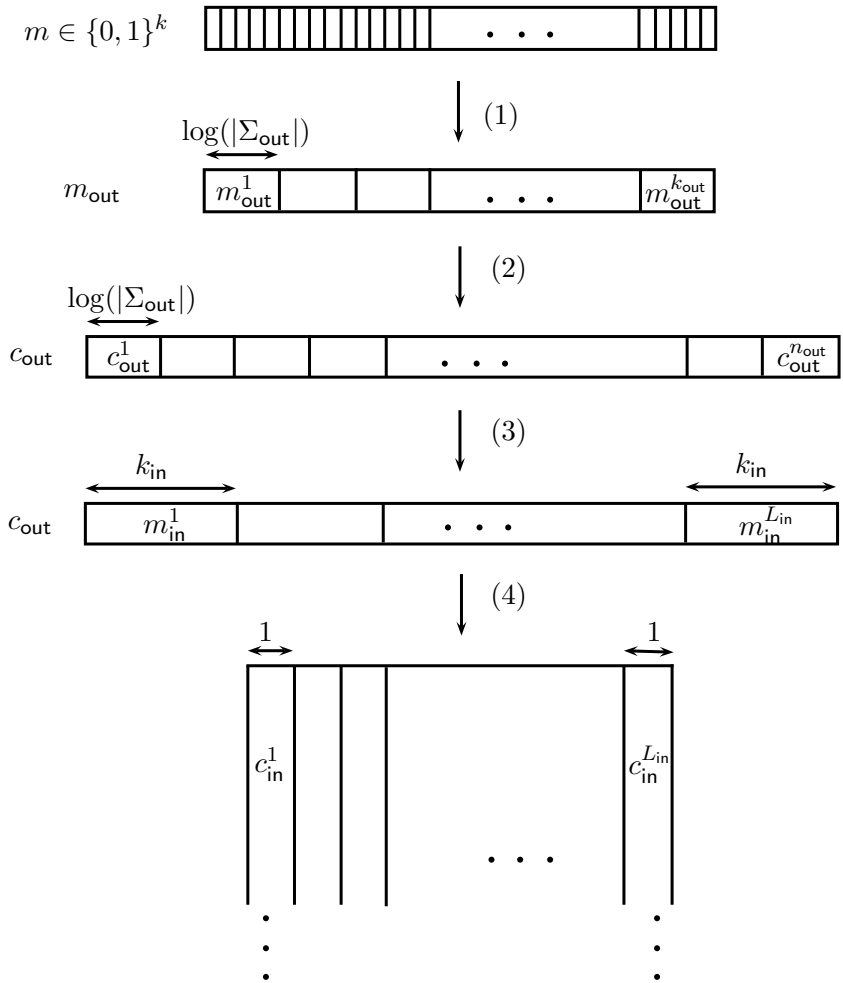

\begin{figure}[htbp]
\centering
\scalebox{1} 
{
\begin{pspicture}(0,-8.2)(14.302813,8.22)
\psframe[linewidth=0.04,dimen=outer](10.160937,6.82)(4.4409375,3.0)
\psline[linewidth=0.04cm](5.0809374,6.8)(5.0809374,3.02)
\psline[linewidth=0.04cm](5.7209377,6.8)(5.7209377,3.04)
\psline[linewidth=0.04cm](6.3409376,6.82)(6.3409376,3.02)
\psline[linewidth=0.04cm](9.500937,6.78)(9.500937,3.0)
\usefont{T1}{ptm}{m}{n}
\rput(4.7923436,4.71){$\hat{c}_\inn^1$}
\psline[linewidth=0.04cm,arrowsize=0.05291667cm 2.0,arrowlength=1.4,arrowinset=0.4]{<->}(4.1009374,6.82)(4.1009374,3.02)
\usefont{T1}{ptm}{m}{n}
\rput(3.7523437,4.71){$\nin$}
\psline[linewidth=0.04cm,arrowsize=0.05291667cm 2.0,arrowlength=1.4,arrowinset=0.4]{->}(7.3009377,8.2)(7.3009377,7.02)
\usefont{T1}{ptm}{m}{n}
\rput(8.192344,7.49){$BSC(p)$}
\psline[linewidth=0.04cm,arrowsize=0.05291667cm 2.0,arrowlength=1.4,arrowinset=0.4]{->}(4.6809373,2.8)(2.7409375,1.86)
\psframe[linewidth=0.04,dimen=outer](3.1209376,1.62)(2.3209374,0.2)
\usefont{T1}{ptm}{m}{n}
\rput(2.6523438,-0.91){$\kin$}
\usefont{T1}{ptm}{m}{n}
\rput(2.5723438,-1.73){$\hat{m}_{\inn}^{1}$}
\psline[linewidth=0.04cm,arrowsize=0.05291667cm 2.0,arrowlength=1.4,arrowinset=0.4]{->}(5.4009376,2.8)(4.8409376,1.78)
\psframe[linewidth=0.04,dimen=outer](5.1209373,1.62)(4.3209376,0.2)
\psline[linewidth=0.04cm,arrowsize=0.05291667cm 2.0,arrowlength=1.4,arrowinset=0.4]{->}(6.1209373,2.82)(6.4209375,1.82)
\psframe[linewidth=0.04,dimen=outer](6.9209375,1.62)(6.1209373,0.2)
\psline[linewidth=0.04cm,arrowsize=0.05291667cm 2.0,arrowlength=1.4,arrowinset=0.4]{->}(9.920938,2.8)(12.260938,1.82)
\psframe[linewidth=0.04,dimen=outer](12.700937,1.62)(11.900937,0.2)
\usefont{T1}{ptm}{m}{n}
\rput(12.324062,1.11){ML}
\usefont{T1}{ptm}{m}{n}
\rput(12.299063,0.71){dec}
\usefont{T1}{ptm}{m}{n}
\rput(12.272344,-0.89){$\kin$}
\usefont{T1}{ptm}{m}{n}
\rput(12.152344,-1.73){$\hat{m}_{\inn}^{\Lin}$}
\usefont{T1}{ptm}{m}{n}
\rput(9.832344,4.71){$\hat{c}_\inn^{\Lin}$}
\psline[linewidth=0.04cm,arrowsize=0.05291667cm 2.0,arrowlength=1.4,arrowinset=0.4]{<->}(4.4809375,7.0)(5.1209373,7.0)
\usefont{T1}{ptm}{m}{n}
\rput(4.7523437,7.31){$1$}
\psline[linewidth=0.04cm,arrowsize=0.05291667cm 2.0,arrowlength=1.4,arrowinset=0.4]{<->}(9.4609375,7.02)(10.080937,7.02)
\usefont{T1}{ptm}{m}{n}
\rput(9.812344,7.33){$1$}
\psframe[linewidth=0.04,dimen=outer](13.300938,-1.38)(1.6609375,-2.04)
\psline[linewidth=0.04cm](3.6809375,-1.42)(3.6809375,-2.02)
\psline[linewidth=0.04cm](5.6809373,-1.38)(5.6809373,-2.0)
\psline[linewidth=0.04cm](11.280937,-1.42)(11.280937,-2.02)
\psline[linewidth=0.04cm,arrowsize=0.05291667cm 2.0,arrowlength=1.4,arrowinset=0.4]{<->}(1.6609375,-1.2)(3.6809375,-1.2)
\psline[linewidth=0.04cm,arrowsize=0.05291667cm 2.0,arrowlength=1.4,arrowinset=0.4]{<->}(11.260938,-1.2)(13.200937,-1.2)
\psdots[dotsize=0.1](7.3009377,4.02)
\psdots[dotsize=0.1](7.6809373,4.02)
\psdots[dotsize=0.1](8.080937,4.02)
\psdots[dotsize=0.1](7.8409376,0.66)
\psdots[dotsize=0.1](8.220938,0.66)
\psdots[dotsize=0.1](8.620937,0.66)
\psdots[dotsize=0.1](9.120937,-1.76)
\psdots[dotsize=0.1](9.500937,-1.76)
\psdots[dotsize=0.1](9.900937,-1.76)
\psline[linewidth=0.04cm](7.7009373,-1.38)(7.7009373,-2.0)
\usefont{T1}{ptm}{m}{n}
\rput(2.7440624,1.13){ML}
\usefont{T1}{ptm}{m}{n}
\rput(2.7190626,0.73){dec}
\usefont{T1}{ptm}{m}{n}
\rput(4.7640624,1.11){ML}
\usefont{T1}{ptm}{m}{n}
\rput(4.7390623,0.71){dec}
\usefont{T1}{ptm}{m}{n}
\rput(6.5640626,1.11){ML}
\usefont{T1}{ptm}{m}{n}
\rput(6.5390625,0.71){dec}
\psline[linewidth=0.04cm,arrowsize=0.05291667cm 2.0,arrowlength=1.4,arrowinset=0.4]{->}(2.7009375,0.04)(2.7009375,-0.58)
\psline[linewidth=0.04cm,arrowsize=0.05291667cm 2.0,arrowlength=1.4,arrowinset=0.4]{->}(4.7009373,0.04)(4.7009373,-0.56)
\psline[linewidth=0.04cm,arrowsize=0.05291667cm 2.0,arrowlength=1.4,arrowinset=0.4]{->}(12.320937,0.04)(12.320937,-0.54)
\psline[linewidth=0.04cm,arrowsize=0.05291667cm 2.0,arrowlength=1.4,arrowinset=0.4]{->}(6.5009375,0.0)(6.5009375,-0.56)
\psframe[linewidth=0.04,dimen=outer](13.340938,-3.38)(1.7009375,-4.04)
\psline[linewidth=0.04cm](3.7209375,-3.42)(3.7209375,-4.02)
\psline[linewidth=0.04cm](5.7209377,-3.38)(5.7209377,-4.0)
\psline[linewidth=0.04cm](11.300938,-3.42)(11.300938,-4.02)
\psdots[dotsize=0.1](9.160937,-3.76)
\psdots[dotsize=0.1](9.540937,-3.76)
\psdots[dotsize=0.1](9.940937,-3.76)
\psline[linewidth=0.04cm](7.7409377,-3.38)(7.7409377,-4.0)
\psline[linewidth=0.04cm](2.7209375,-3.42)(2.7209375,-4.02)
\psline[linewidth=0.04cm](4.7009373,-3.42)(4.7009373,-4.02)
\psline[linewidth=0.04cm](6.7209377,-3.42)(6.7209377,-4.02)
\psline[linewidth=0.04cm](12.300938,-3.42)(12.300938,-4.02)
\usefont{T1}{ptm}{m}{n}
\rput(3.5723438,-5.05){$\log ( |\Sigma_\out |)$}
\psframe[linewidth=0.04,dimen=outer](11.760938,-5.52)(3.1209376,-6.18)
\psline[linewidth=0.04cm](5.1209373,-5.52)(5.1209373,-6.14)
\psline[linewidth=0.04cm](10.700937,-5.56)(10.700937,-6.16)
\psdots[dotsize=0.1](8.560938,-5.9)
\psdots[dotsize=0.1](8.940937,-5.9)
\psdots[dotsize=0.1](9.340938,-5.9)
\psline[linewidth=0.04cm](7.1409373,-5.52)(7.1409373,-6.14)
\psline[linewidth=0.04cm](4.1009374,-5.56)(4.1009374,-6.16)
\psline[linewidth=0.04cm](6.1209373,-5.56)(6.1209373,-6.16)
\psframe[linewidth=0.04,dimen=outer](11.760938,-7.54)(3.1209376,-8.2)
\psline[linewidth=0.04cm](5.1209373,-7.54)(5.1209373,-8.16)
\psline[linewidth=0.04cm](10.700937,-7.58)(10.700937,-8.18)
\psdots[dotsize=0.1](8.560938,-7.92)
\psdots[dotsize=0.1](8.940937,-7.92)
\psdots[dotsize=0.1](9.340938,-7.92)
\psline[linewidth=0.04cm](7.1409373,-7.54)(7.1409373,-8.16)
\psline[linewidth=0.04cm](4.1009374,-7.58)(4.1009374,-8.18)
\psline[linewidth=0.04cm](6.1209373,-7.58)(6.1209373,-8.18)
\psline[linewidth=0.04cm](3.3209374,-7.56)(3.3209374,-8.16)
\psline[linewidth=0.04cm](3.5209374,-7.58)(3.5209374,-8.18)
\psline[linewidth=0.04cm](3.7209375,-7.58)(3.7209375,-8.18)
\psline[linewidth=0.04cm](3.9009376,-7.58)(3.9009376,-8.18)
\psline[linewidth=0.04cm](4.3209376,-7.56)(4.3209376,-8.16)
\psline[linewidth=0.04cm](4.5209374,-7.58)(4.5209374,-8.18)
\psline[linewidth=0.04cm](4.7209377,-7.58)(4.7209377,-8.18)
\psline[linewidth=0.04cm](4.9009376,-7.58)(4.9009376,-8.18)
\psline[linewidth=0.04cm](5.3209376,-7.54)(5.3209376,-8.14)
\psline[linewidth=0.04cm](5.5209374,-7.56)(5.5209374,-8.16)
\psline[linewidth=0.04cm](5.7209377,-7.56)(5.7209377,-8.16)
\psline[linewidth=0.04cm](5.9009376,-7.56)(5.9009376,-8.16)
\psline[linewidth=0.04cm](6.3409376,-7.56)(6.3409376,-8.16)
\psline[linewidth=0.04cm](6.5409374,-7.58)(6.5409374,-8.18)
\psline[linewidth=0.04cm](6.7409377,-7.58)(6.7409377,-8.18)
\psline[linewidth=0.04cm](6.9209375,-7.58)(6.9209375,-8.18)
\psline[linewidth=0.04cm](10.920938,-7.56)(10.920938,-8.16)
\psline[linewidth=0.04cm](11.120937,-7.58)(11.120937,-8.18)
\psline[linewidth=0.04cm](11.320937,-7.58)(11.320937,-8.18)
\psline[linewidth=0.04cm](11.500937,-7.58)(11.500937,-8.18)
\psline[linewidth=0.04cm,arrowsize=0.05291667cm 2.0,arrowlength=1.4,arrowinset=0.4]{->}(7.3209376,-6.42)(7.3209376,-7.38)
\psline[linewidth=0.04cm,arrowsize=0.05291667cm 2.0,arrowlength=1.4,arrowinset=0.4]{->}(7.3209376,-4.42)(7.3209376,-5.34)
\usefont{T1}{ptm}{m}{n}
\rput(7.8323436,-6.87){$(1)$}
\usefont{T1}{ptm}{m}{n}
\rput(7.8923435,-4.83){$(2)$}
\psline[linewidth=0.04cm,arrowsize=0.05291667cm 2.0,arrowlength=1.4,arrowinset=0.4]{->}(7.3209376,-2.24)(7.3209376,-3.2)
\usefont{T1}{ptm}{m}{n}
\rput(7.8923435,-2.69){$(3)$}
\psline[linewidth=0.04cm,arrowsize=0.05291667cm 2.0,arrowlength=1.4,arrowinset=0.4]{<->}(1.7209375,-3.18)(2.6809375,-3.18)
\usefont{T1}{ptm}{m}{n}
\rput(2.1723437,-2.91){$\log ( |\Sigma_\out |)$}
\usefont{T1}{ptm}{m}{n}
\rput(7.092344,2.37){$(4)$}
\usefont{T1}{ptm}{m}{n}
\rput(2.2223437,-3.71){$\hat{c}_{\out}^{1}$}
\usefont{T1}{ptm}{m}{n}
\rput(12.822344,-3.73){$\hat{c}_{\out}^{\nout}$}
\usefont{T1}{ptm}{m}{n}
\rput(0.77234375,-3.75){$\hat{c}_{\out}$}
\usefont{T1}{ptm}{m}{n}
\rput(11.262343,-5.87){$\hat{m}_{\out}^{\kout}$}
\usefont{T1}{ptm}{m}{n}
\rput(3.6223438,-5.85){$\hat{m}_{\out}^{1}$}
\usefont{T1}{ptm}{m}{n}
\rput(1.7523438,-5.91){$\hat{m}_{\out}$}
\usefont{T1}{ptm}{m}{n}
\rput(1.8723438,-7.91){$\hat{m}\in \BIT{k}$}
\psline[linewidth=0.04cm,arrowsize=0.05291667cm 2.0,arrowlength=1.4,arrowinset=0.4]{<->}(3.1009376,-5.4)(4.0609374,-5.4)
\usefont{T1}{ptm}{m}{n}
\end{pspicture} 
}
\caption{Decoder of the concatenated code uses ML-decoding for the inner code and the
decoder of the outer code}
\label{figure:dec}
\end{figure}
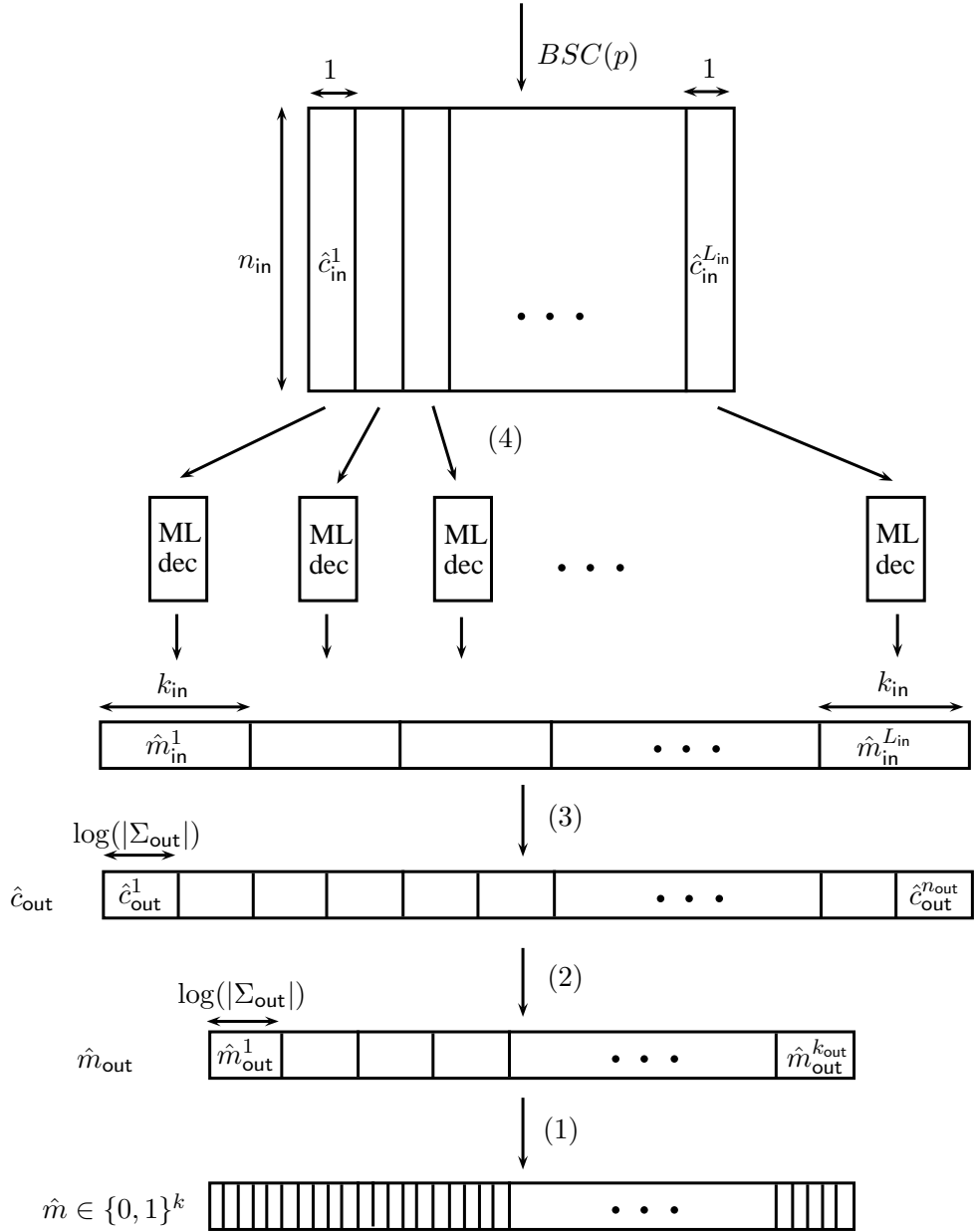


\remove{
\section{Discussion}
\begin{enumerate}
\item Emphasize that our code is systematic.
\item Does the decoder know the crossover probability?
\item Extension to other DMCs? Gaussian?
\item Extension to BEC?
\item Extension to adversarial channel?
\item Comparison to previous rateless code (spinal codes): deterministic, proven
  without assumption about good hash function. parallel encoding and decoding.
\item General question of existence of a specific infinite object that exists at random.
\item Apart from achieving capacity, one may be interested in having
  $\lim_{n\to\infty}\er(p,k,n)=0$, for every sufficiently large $k$.  Such a result
  requires designing an inner rateless code, the minimum distance of which increases
  as a function of the codeword length.
\item Strengthen theorem~\ref{thm:inner}: Prove that the inefficient rateless code can also have \[\er(p,k,n)=e^{-\Omega(k/\log (k))}\:.\]
\end{enumerate}
}
{\paragraph{Acknowledgments.} We thank Uri Erez, Meir Feder, Simon Litsyn, and Rami Zamir for useful conversations.}


\bibliographystyle{alpha}
\bibliography{rateless}


\appendix
\section{Omitted Proofs}

\subsection{Proof of Lemma~\ref{lem:Rn}}\label{app:Rn}
We begin with the following claim.
\begin{myclaim}\label{claim:split}
  For every set $W\subseteq \bit^k\setminus \{0^k\}$ of size at least $2n^2$, there
  are more than  $2^k \cdot (1-\frac{1}{2n})$ vectors that
 $\frac{1}{2\cdot \sqrt{n}}$-split $W$.
\end{myclaim}
\begin{proof}
  Let $W=\{x_1,\ldots, x_m\}$, where $m\geq 2n^2$.  Let $R$ denote a random vector
  chosen uniformly in $\BIT{k}$.  This uniform distribution induces $m$ random
  variables defined by
\begin{align*}
  Z_i & \triangleq
  \begin{cases}
    1 & \text{if $R\cdot x_i=1$,}\\
    0 & \text{if $R\cdot x_i=0$.}
  \end{cases}
\end{align*}
The expectation of each random variable $Z_i$ is $1/2$, and the variance of each
$Z_i$ is $1/4$. (However, they are not independent.) Since the elements of $W$ are
distinct, the random variables $\{Z_i\}_i$ are pairwise independent.
By Chebyshev's Inequality,
\begin{align}
  \label{eq:Cheb}
\Prob \left( \left|\frac 1m\cdot \sum_{i=1}^m Z_i -\frac 12 \right| > \frac{1}{2\cdot\sqrt{n}}\right) < \frac{1}{2n}.
\end{align}
To complete the proof, note that $R$ is an $\frac{1}{2\cdot \sqrt{n}}$-splitter for $W$
if and only if
$\left|\frac 1m\cdot \sum_{i=1}^m Z_i -\frac 12 \right| \leq
\frac{1}{2\cdot\sqrt{n}}$.
\myqed\end{proof}

\begin{proof}[Proof of Lemma~\ref{lem:Rn}]
  If $|W_{i,n}\setminus M|<2n^2$, then all the information words in $W_{i,n}$ are
  marked, and $\widehat{W_{i,n}}$ is empty. Therefore, $\widehat{W}_{i,n}$ is either
    empty or of size at least $2n^2$. It follows, by a union bound, that more than
    half of the $k$-bit vectors simultaneously $\frac{1}{2\cdot \sqrt{n}}$-split each
    set $\widehat{W}_{i,n}$, for $1\leq i\leq n$.  Therefore, to prove the lemma it
    suffices to show that at least half of the $R$'s $(1/8)$-elevates the set
    $W_{d,n}$.

    Note that any $3/8$-splitter of $W_{d,n}$ is also a $1/8$-elevator of this set.
    In case $|W_{d,n}| \leq 8$, pick a vector $x\in W_{d,n}$. Half the vectors are
    not orthogonal to $x$, and hence at least half the vectors are $1/8$-elevators of
    $W_{d,n}$.  If $|W_{d,n}| >9$, we can apply the argument of the above
    Claim~\ref{claim:split} and get that at least half of the $R$'s $1/8$-elevate
    $W_{d,n}$.  This completes the proof of Lemma~\ref{lem:Rn}.
\end{proof}
\subsection{Bound on number of unmarked vectors}\label{appendix:bound}
\begin{myclaim}\label{claim:f}
  \begin{align*}
    |\widehat{W}_{i,k+t}| & \leq \left (\frac{2}{3} \right )^{t} \binom {k+t}{i} \:.
  \end{align*}
\end{myclaim}
\begin{proof}
  The proof is by induction on $t$.  The induction basis for $t=0$ holds because
  $|W_{i,k}|=\binom{k}{i}$.
We now prove the induction step for $t+1$. The choice of $R_{k+t+1}$ splits each
$\widehat{W}_{i,k+t}$
so that
\begin{align*}
   |\widehat{W}_{i,k+t+1}|&\leq \frac{2}{3} \cdot \left (|\widehat{W}_{i-1,k+t}|+|\widehat{W}_{i,k+t}|\right).
\end{align*}
The induction hypothesis for $t$ implies that
\begin{align*}
   |\widehat{W}_{i,k+t+1}|&\leq \left(\frac{2}{3}\right)^{t+1} \cdot \left (  \binom
     {k+t}{i-1} + \binom {k+t}{i}\right)\\
&=\left(\frac{2}{3}\right)^{t+1} \cdot \binom {k+t+1}{i},
\end{align*}
and the claim follows.
\myqed
\end{proof}

\subsection{Proof of Extension of Poltyrev's Theorem}\label{sec:poltyrev}\label{appendix:P}
Before we prove Theorem~\ref{thm:poltyrev}, we collect some useful facts.
The extension of the binomial coefficients to reals is defined by
  \begin{align}\label{Eq:gamma}
    \binom{n}{k}=\frac{\Gamma(n+1)}{\Gamma(k+1)\Gamma(n-k+1)},
  \end{align}
where $\Gamma(x)$ is the Gamma function that extends the factorial function to the real numbers.
In particular, $\Gamma(x)$ is monotone increasing for $x\geq 1$ and $\Gamma(x+1)=x\cdot \Gamma(x)$.

\begin{lemma}\label{lemma:comb}
  Let $0<a\leq b$. Define the function $f:[0,b] \rightarrow \RR$ by
  \begin{align*}
    f(x) \triangleq \binom{a}{x}\binom{b}{x}.
  \end{align*}
If $a^2\geq a+b$, then
\begin{align}
\left|\argmax f(x) - \frac{ab-1}{a+b+2}\right| \leq 1,\\
\max\{ f(x)\}  \leq f\left(\frac{ab}{a+b}\right) \cdot \frac{(ab)^4}{(a+b)^4}.
  \end{align}
\end{lemma}
The following proof is based on~\cite{Barg_enee739lecture4}.
\begin{proof}
Let $t\triangleq \frac{ab-1}{a+b+2}$. We first prove the following ``discrete'' monotonicity property:
  \begin{enumerate}
  \item If $i\leq t$, then $f(i)\leq f(i+1)$
  \item If $i\geq t$, then $f(i)\geq f(i+1)$
  \end{enumerate}
The proof of this monotonicity property is by evaluating the quotient
  \begin{align*}
    Q \triangleq
    \frac{f(i)}{f(i+1)}=\frac{\binom{a}{i}\binom{b}{i}}{\binom{a}{i+1}\binom{b}{i+1}}.
  \end{align*}
It is easy to check that $Q\leq 1$ if $i\leq t$, and $Q\geq 1$ if $i\geq t$.

Let $x^*\triangleq \argmax f(x)$. The monotonicity property implies that $t\leq x^*
\leq t+1$. Let $y\triangleq ab/(a+b)$, which proves the first part of the lemma.

Note that $y$ is also between $t$ and $t+1$. This implies that $|x^*-y|\leq 1$.
Note also that $y\geq 1$, $(a-y)=a^2/(a+b)\geq 1$, and $(b-y)=b^2/(a+b)\geq 1$.
Hence by the properties of the Gamma function we obtain:
\begin{align*}
  \frac{f(x^*)}{f(y)} &
= \frac{\Gamma^2(y+1) \cdot \Gamma(a+1-y) \cdot \Gamma(b+1-y)}
{\Gamma^2(x^*+1) \cdot \Gamma(a+1-x^*) \cdot \Gamma(b+1-x^*)}\\
&\leq \frac{\Gamma^2(y+1) \cdot \Gamma(a+1-y) \cdot \Gamma(b+1-y)}
{\Gamma^2(y) \cdot \Gamma(a-y) \cdot \Gamma(b-y)}\\
&=y^2 \cdot (a-y) \cdot (b-y)\\
&=\left(\frac{ab}{a+b}\right)^2 \cdot \frac{a^2}{a+b}\cdot\frac{b^2}{a+b}\\
&= \frac{(ab)^4}{(a+b)^4}.
\end{align*}
\end{proof}

\begin{definition}
Let $\delta_{\GV}(n,k)$ be the root $x\in(0,1/2)$ of the equation $H(x)=1-\frac kn$.
\end{definition}
\begin{definition}
  Let $w^*_i(n,k)\triangleq \binom{n}{i}\cdot 2^{k-n}$ denote the expected weight
  distribution of a random linear $[n,k]$ code.
\end{definition}

\begin{theorem}[refinement of Thm.~1 of~\cite{poltyrev1994bounds}]
Let $p\in (0,\half)$ be a constant, $\delta >0$ be a constant such that $\frac{k}{n}<1-H(p)-\delta$, and $\tau\in [0,1]$ be a threshold parameter. There exists a constant $\alpha>0$ for which the following holds. If $C$ is an $[n,k]$ linear code whose weight distribution $\{w_{i}(C_n)\}_i$
     satisfies \[w_i \leq      2^{(\delta / 3) n} \cdot w^*_i(n,k) \qquad \text{for every } i\geq \tau n.\]
  Then, the probability over $\BSC(p)$ that the all zero word is ML-decoded to a codeword of weight at least $\tau n$ is $2^{-\alpha n}$.
\end{theorem}

\begin{proof}
  Let $y$ be the received word when the all zero word is transmitted (i.e,
  $\{y_i\}_i$ are independent Bernoulli variables with probability $p$).  Let
  $\hat{y}$ denote the codeword computed by the ML-decoder with respect to the input
  $y$.  Our goal is to upper-bound the event that $\hat{y}$ has weight at least $\ell\triangleq \tau n$.

Let $\epsilon>0$  denote a sufficiently small constant that depends only on $p$ and $\delta$; in particular $\epsilon$ satisfies:
\begin{align}
  \label{eq:eps}
  \epsilon &\leq \min \{\frac 12 - p, p\}.
\end{align}

We divide the analysis into two cases based on the
  Hamming weight of $y$:
  \begin{align*}
    \Pr_{y\rs \BSC(p)}(\wt(\hat{y}) \geq \ell) &\leq
    \Pr_{y\rs \BSC(p)}[ | \wt(y)-np|> \epsilon n]
    \\
    &+ \Pr_{y\rs \BSC(p)}[ \wt(\hat{y})\geq \ell \text{ }\&\text{ } | \wt(y)-np| \leq
    \epsilon n]
  \end{align*}

\paragraph{Case 1: The weight of $y$ is far from $np$, i.e $| \wt(y)-np|> \epsilon
  n$.}
By additive Chernoff - Heoffding inequality we know that,
\begin{align*}
   \Pr[|\wt(y)-np|>\epsilon n] \leq 2\cdot e^{-2 \epsilon^2\cdot n} = 2^{-\Omega( n)}.
    \end{align*}

\paragraph{Case 2: The weight of $y$ is close to $np$, i.e $| \wt(y)-np| \leq \epsilon n$.}

Let $r \triangleq \wt(y)$.
Note that,
\begin{align}\label{eq:r}
 pn -\epsilon n \leq r \leq pn + \epsilon n.
\end{align}
Let $ P_{\ell,r}$ denote the following probability
\begin{align*}
  P_{\ell,r}\triangleq \Pr_{y\rs \BSC(p)}[ \wt(\hat{y})\geq \ell \text{ }\&\text{ }  \wt(y)= r] .
\end{align*}
Because all $y$'s of weight $r$ are equiprobable, we have
\begin{align*}
  \Pr_{y\rs \BSC(p)}[ \wt(\hat{y})\geq \ell \mid  \wt(y)= r] &= \frac{ |\{ y: \wt(y)=r ,\wt(\hat{y}) \geq \ell \}|}{|\{y:\wt(y)=r\}|}.
\end{align*}
Hence,
  \begin{align} \nonumber
  P_{\ell,r}&=\Pr_{y\rs \BSC(p)}[ \wt(y)= r] \cdot
\frac{ |\{ y: \wt(y)=r ,\wt(\hat{y}) \geq \ell \}|}{|\{y:\wt(y)=r\}|}.
\\
&\leq
\sum_{i=\ell}^{n} \sum_{c \in C: \wt(c)=i} \frac{|\{y:\wt(y)=r,\hat{y}=c \}|}{|\{y:\wt(y)=r\}|}\;.
\label{Eq:UB}
  \end{align}
Let
  \begin{align*}
     \alpha_{c,r} \triangleq | \{y:\hat{y}=c ~\&~ \wt(y)=r \}|.
  \end{align*}

  Fix a codeword $c\in C$ of weight $i$.  A word $y$ of weight $r$ is ML-decoded to
  $c$ only if $\dist(y,c) \leq r$.
Without loss of generality $c=1^i\circ 0^{n-i}$ (i.e., $c$ consists of $i$ ones
followed $n-i$ zeros).
Note that $\wt(y)=\dist(y,0^n)$. Let $y'$ and $y"$ denote the prefix of length $i$ of
$y$ and the suffix of length $n-i$ of $y$, respectively.
Because $\dist(y,c) \leq r$, it follows that $0\leq r-\dist(y,c) =
\dist(y,0^n)-\dist(y,c)$.
But
\begin{align*}
  \dist(y,0^n)-\dist(y,c) &= \dist(y',0^i) + \dist(y",0^{n-i})
-\dist(y',1^i) + \dist(y",0^{n-i})\\
&=\dist(y',0^i)-\dist(y',1^i).
\end{align*}
Namely, in the prefix $y'$, the majority of the bits are ones. We conclude that at
least $i/2$ of the coordinates of the support $y$ have to be chosen from the
coordinates of the support of $c$.  Hence,
\begin{align}\label{Eq:given}
 \alpha_{c,r} &=\sum_{w=i/2}^{r} \binom{i}{w} \binom{n-i}{r-w}.
\end{align}

Because, $\binom{i}{w} \leq \binom{i}{i/2}$, we can upper-bound (\ref{Eq:given}) by,
\begin{align*}
 \alpha_{c,r} \leq \binom{i}{i/2}\sum_{w=0}^{r-i/2} \binom{n-i}{w}.
\end{align*}
Because $\epsilon\leq\frac{1}{2}-p$ the maximal summand is $\binom{n-i}{r-i/2}$, and we get an upper-bound of
\begin{align}\label{Eq:alpha}
 \alpha_{c,r} \leq n \binom{i}{i/2} \binom{n-i}{r-i/2}.
\end{align}

Substituting Eq.~(\ref{Eq:alpha}) in Eq.~(\ref{Eq:UB}), we get,


\begin{align*}
P_{\ell,r} \leq  \sum_{i=\ell}^{n} w_{i,n}\frac{n \binom{i}{i/2} \binom{n-i}{r-i/2}}{\binom{n}{r}}\;.
\end{align*}

The weight distribution $w_{i,n}$ satisfies
$w_{i,n} = 2^{(\delta/3)n} \cdot w^*_i(n,k)$, therefore,
\begin{align*}
P_{\ell,r}\leq  n\binom{n}{r}^{-1}\sum_{i=\ell}^{n} 2^{(\delta/3)n} \cdot w^*_i(n,k) \binom{i}{i/2} \binom{n-i}{r-i/2}\;.
\end{align*}

Recall that the average weight distribution $w^*_i(n,k) $ satisfies
\begin{align*}
w^*_i(n,k)=2^{k-n}\binom{n}{i}
\end{align*}

therefore,
\begin{align}\label{Eq:UB2}
P_{\ell,r} \leq  n\binom{n}{r}^{-1}\sum_{i=\ell}^{n} 2^{k-n+(\delta/3)n}\binom{n}{i} \binom{i}{i/2} \binom{n-i}{r-i/2}\;.
\end{align}

Now, we show that,
\begin{align}\label{Eq:id}
\binom{n}{i} \binom{i}{i/2} \binom{n-i}{r-i/2}=\binom{n}{r}  \binom{r}{i/2}\binom{n-r}{i/2}.
\end{align}

The combinatorial proof proceeds by counting the number of possibilities of dividing
students to two classes and choosing committee members in two ways.  Consider $n$
students that we wish to partition to two classes one of size $i$ and the other of
size $n-i$. We want to choose a committee of $r$ students that consists of $i/2$
students from the first class, and $r-i/2$ students from the second class.
The left hand side in Eq.~\ref{Eq:id} counts the number of possible partitions into
two classes and choices of committee members as follows.
First partition the students by choosing the members of the first class, then choose
the committee members from each class.
The right hand side in Eq.~\ref{Eq:id} counts the same number of possibilities by
first choosing the committee members (before dividing the students into
classes). Only then we partition the committee members to two classes. Finally, the
non-committee members  of the first class are chosen.

Plugging in (\ref{Eq:id}) and (\ref{Eq:UB2}), we get,
\begin{align*}
 P_{\ell,r} \leq n\sum_{i=\ell}^{n} 2^{k-n+(\delta/3)n}   \binom{r}{i/2} \binom{n-r}{i/2}
\end{align*}

By Lemma~\ref{lemma:comb},
\begin{align*}
  \binom{r}{i/2} \binom{n-r}{i/2} &\leq \binom{r}{r(n-r)/n} \binom{n-r}{r(n-r)/n}
  \cdot \left( \frac{r(n-r)}{n}\right)^4\\
&\leq  \binom{r}{r(n-r)/n} \binom{n-r}{r(n-r)/n}  \cdot n^4.
\end{align*}

It follows that
\begin{align}\label{Eq:error1}
 P_{\ell,r} &\leq n^6 \cdot 2^{k-n+(\delta/3)n} \cdot \binom{r}{r(n-r)/n} \binom{n-r}{r(n-r)/n}.
\end{align}
Let $\hat{p}\triangleq\frac{r}{n}$. By Eq.~\ref{eq:r} it follows that
\begin{align}\label{Eq:error3}
 P_{\ell,r} &\leq n^6\cdot 2^{k-n+(\delta/3)n}  \cdot \binom{\hat{p}n}{\hat{p}(1-\hat{p})n} \binom{(1-\hat{p})n}{\hat{p}(1-\hat{p})n}.
\end{align}

Because
\begin{align*}
  \binom{n}{k} \leq 2^{nH(\frac{k}{n})},
\end{align*}
it follows that
\begin{align*}
P_{\ell,r} \leq n^6 \cdot 2^{k-n+(\delta/3)n} \cdot  2^{\hat{p}nH(1-\hat{p})} \cdot  2^{(1-\hat{p})nH(\hat{p})}
\end{align*}
Because $H(\hat{p})=H(1-\hat{p})$, we get,
\begin{align*}
 P_{\ell,r} \leq n^6 \cdot  2^{k-n+(\delta/3)n +nH(\hat{p})}
\end{align*}
Our goal now is to prove that the exponent $k-n+(\delta/3)n +nH(\hat{p})$ is
at most $-\delta \cdot n/3$.
Indeed,
\begin{align*}
k-n+(\delta/3)n +nH(\hat{p}) &= -n\cdot (-\frac kn+1 -\frac{\delta}{3} - H(\hat{p}))
\\
&\leq -n \cdot(H(p)- H(\hat{p})+\frac{2}{3}\cdot \delta ) .
  \end{align*}
  To complete the proof, it suffices to show that $|H(\hat{p})-H(p)|< \delta/3$.
Indeed, $|p-\hat{p}|\leq \epsilon$, and hence by
  continuity, this holds if $\epsilon$ is sufficiently small (as a function of $p$ and
  $\delta)$.
\end{proof}

\end{document}